\documentclass[a4paper,UKenglish,cleveref, autoref, thm-restate]{lipics-v2021}

\usepackage{mathtools}
\usepackage{todonotes}
\usepackage{amsthm}
\usepackage{placeins}
\usepackage{algorithm}
\usepackage[noend]{algpseudocode}
\usepackage{xcolor}
\usepackage{subcaption}
\usepackage{adjustbox}

\usepackage{tikz}
\usetikzlibrary{calc}
\usetikzlibrary{arrows.meta}
\usetikzlibrary {positioning}
\usetikzlibrary{decorations.pathreplacing,angles,quotes}

\pdfoutput=1 %
\hideLIPIcs  %

\title{Advances in Exact and Approximate Group Closeness Centrality Maximization}

\titlerunning{Exact and Approximate Group Closeness Maximization}

\author{Christian Schulz}{Heidelberg University, Faculty of Mathematics and Computer Science, Germany}{christian.schulz@informatik.uni-heidelberg.de}{https://orcid.org/0000-0002-2823-3506}{}
\author{Jakob Ternes\footnote{Corresponding author}}{Heidelberg University, Faculty of Mathematics and Computer Science, Germany}{jakob.ternes@student.kit.edu}{https://orcid.org/0009-0003-2857-8293}{}
\author{Henning Woydt}{Heidelberg University, Faculty of Mathematics and Computer Science, Germany}{henning.woydt@informatik.uni-heidelberg.de}{https://orcid.org/0009-0004-2234-2869}{Supported by the Deutsche Forschungsgemeinschaft (DFG, German Research Foundation) -- DFG SCHU 2567/6-1}

\authorrunning{C. Schulz, J. Ternes and H. Woydt}

\Copyright{Christian Schulz, Jakob Ternes and Henning Woydt}



\category{} %

\supplement{}
\supplementdetails[subcategory={Source Code}, cite={}, swhid={}]{Software}{https://github.com/AEGH-Grover/Grover}

\acknowledgements{The results of this work are based on the corresponding author’s bachelor's thesis \cite{ternes_2025}.}%

\nolinenumbers %

\newcommand{\dist}{\text{dist}}
\newcommand{\ecc}{\text{ecc}}
\newcommand{\diam}{\text{diam}}

\newcommand{\gccm}{\textsc{GCCM}}
\newcommand{\greedy}{\textsc{Greedy}}

\begin{document}

    \maketitle

    \begin{abstract}
        In the NP-hard \textsc{Group Closeness Centrality Maximization} problem, the input is a graph $G = (V,E)$ and a positive integer $k$, and the task is to find a set $S \subseteq V$ of size $k$ that minimizes group farness $f(S) = \sum_{v \in V} \min_{s \in S}\dist(v,s)$.
        The state-of-the-art exact algorithm iteratively solves ILPs of increasing size until the final ILP can provably represent an optimal solution.
        We introduce a new data reduction technique that eliminates variables from the ILP by proving that certain vertices have their distance to any optimal solution structurally determined by a neighbor.
        Additionally, we bootstrap the exact solver with an approximate solution to produce near-sufficient ILPs from the first iteration, reducing the number of needed iterations.
        Our improvements yield a speedup by a factor of $4.5$ over the next best exact algorithm and can achieve speedups by up to a factor of $34.1$.
        Furthermore, we add reduction techniques to a $1/5$-approximation algorithm, and show that these adaptations do not compromise its approximation guarantee.
        The improved algorithm achieves mean speedups of up to $1.6$ and a maximum speedup of $9.6$ times.
        Finally, we settle an open question by proving that a widely used greedy algorithm admits arbitrarily poor approximation ratios.
    \end{abstract}

    \newpage
    \setcounter{page}{1}

    \section{Introduction}\label{sec:introduction}
    A crucial task in the analysis of networks is the identification of vertices or groups of vertices which, in some sense, are important.
    Centrality measures quantify the importance by assigning a numeric value to a vertex or a group of vertices.
    They are frequently used in the analysis of social networks~\cite{social2013}, biological networks~\cite{van2013network, chea2007accurate}, electrical power grids~\cite{AdebayoSun+2020} and, among many other uses, in facility location problems~\cite{Koschuetzki2005}.
    Various measures have been proposed to identify the centrality of an individual vertex in a network.
    For instance, closeness centrality measures how far a node is from the other nodes in the graph, betweenness centrality measures how often a node lies on the shortest path between other nodes, and degree centrality simply captures the number of neighbors.
    As noted by Everett and Borgatti~\cite{everett1999centrality}, these measures fall short of identifying vertices which are collectively central as a group, leading them to propose generalizations of the aforementioned measures to groups of vertices.
    Indeed, it has been shown that the overlap between the top $k$ nodes with the greatest individual closeness centrality and groups of size $k$ with high group closeness centrality is relatively small~\cite{DBLP:journals/corr/abs-1710-01144}.
    Naturally, the question arises of how to efficiently find groups of a given size $k$ such that a certain group centrality measure is maximized.
    In the NP-hard~\cite{chen2016efficient, staus2023exact} group closeness centrality maximization (\gccm) problem, a positive integer $k$ and an undirected connected graph $G = (V,E)$ are given, and the task is to find a group \hbox{$S \subseteq V$} of size $k$ that maximizes
    \[
        c(S) \coloneq \frac{|V|-k}{\sum_{v \in V} \min_{s \in S} \dist(v,s)}.
    \]

    There are two prominent ways to exactly solve the problem.
    On the one hand, branch-and-bound algorithms specifically designed for this problem are used.
    These algorithms exhaustively explore a search-tree containing all possible solutions.
    To speed up the computation, many pruning techniques are used to cut off large parts of the explorable search-space.
    Such methods have been explored by Rajbhandari~et~al.~\cite{rajbhandari2022presto}, Staus~et~al.~\cite{staus2023exact} and Woydt~et~al.~\cite{Woydt24}.
    On the other hand, ILP solvers can be used to solve \gccm.
    Since ILPs can capture a vast range of problems, a huge amount of research and engineering has flowed into ILP solvers.
    This has made ILP solvers an incredibly powerful tool for a large class of optimization problems, including \gccm.
    Bergamini~et~al.~\cite{DBLP:journals/corr/abs-1710-01144} were the first to propose an ILP for this problem with $\mathcal{O}(n^2)$ variables and constraints.
    Staus~et~al.~\cite{staus2023exact} then improved on the proposed formulation and reduced the variable and constraint count to $\mathcal{O}(n \cdot \diam(G))$.
    Additionally, they proposed an iterative approach that starts by solving smaller ILPs and only if necessary enlarges them.
    Experimental evaluations in~\cite{staus2023exact, Woydt24} show that ILP solvers currently outperform the branch-and-bound approach.

    When solution quality is not paramount, one may decide to use a heuristic or approximation algorithm.
    Chen~et~al.~\cite{chen2016efficient} considered a simple greedy algorithm that starts with the empty set and iteratively adds the vertex that improves the score the most.
    They claimed that the algorithm has a $(1-1/e)$-approximation guarantee.
    The runtime of the algorithm was massively improved by Bergamini~et~al.~\cite{DBLP:journals/corr/abs-1710-01144} and they observed very high empirical solution quality when compared to an optimal solution.
    However, they also showed that the $(1-1/e)$-approximation guarantee was erroneously assumed.
    It is therefore unclear whether the simple greedy algorithm has an approximation guarantee for \gccm.
    Angriman~et~al.~\cite{angriman2021group} later presented an algorithm with a $1/5$-approximation ratio that is based on a local search algorithm for the related $k$-median problem~\cite{arya2001local}.
    The algorithm may start with any initial solution, which is then refined using a swap-based local search procedure.
    Vertices inside the solution are swapped with vertices outside of the solution if the swap improves the objective.
    If no more objective-improving swaps can be made then the algorithm terminates.

    \textbf{Our Contribution.}
    We develop new algorithmic techniques for group closeness centrality maximization, both in the exact and approximate case.
    In the exact case, we introduce \emph{absorbed vertices}, a data reduction that goes beyond the known dominated vertex removal proposed by Staus~et~al.~\cite{staus2023exact}.
    Additionally, we bootstrap the iterative ILP framework with an approximate solution, producing near-optimal ILPs from the first iteration.
    For some instances, these improvements can lead to a speedup of $34.16\times$, while the overall mean speedup is $4.59\times$.
    In the approximate case, we prove that the search space of the $1/5$-approximation algorithm of Angriman~et~al.~\cite{angriman2021group} can be restricted to non-dominated vertices without compromising its approximation guarantee, achieving mean speedups of $1.61\times$.
    Finally, we show that the greedy algorithm proposed by Chen~et~al.~\cite{chen2016efficient} does not have an approximation guarantee for \gccm.

    \section{Preliminaries}\label{sec:preliminaries}
    Let \( G = (V, E) \) be an undirected, unweighted, connected graph without self-loops.
    Two vertices \( u, v \in V \) are \emph{adjacent} if \( \{u,v\} \in E \).
    The \emph{neighborhood} of a vertex \( v \) is \( N(v) = \{u \in V \mid \{u,v\} \in E\} \), and its \emph{degree} is \( \deg(v) = |N(v)| \).
    The \emph{maximum degree} is \( \Delta(G) = \max_{v \in V} \deg(v) \).
    The \emph{closed neighborhood} is \( N[v] = N(v) \cup \{v\} \), and \( v \) \emph{dominates} \( u \) if \( N[u] \subseteq N[v] \).
    The \emph{distance} \( \dist(u,v) \) is the length of a shortest \(u\)-\(v\) path, and for a set \( S \subseteq V \), \( \dist(u,S) = \min_{v \in S} \dist(u,v) \).
    The \emph{eccentricity} of \( v \) is \( \ecc(v) = \max_{u \in V} \dist(u,v) \), and the \emph{diameter} of \( G \) is \( \diam(G) = \max_{v \in V} \ecc(v) \).
    A vertex \( w \in V \) is a \emph{cut vertex} if removing it increases the number of connected components.

    The \emph{closeness centrality} measure~\cite{10f5227e-a114-31a3-b1b2-cc4f421e84a5, https://doi.org/10.1002/bs.3830100205, sabidussi1966centrality} (see Eq.~\ref{eq:closeness-vertex}) scores each vertex $v \in V$ based on how close it is to all other vertices, while the \emph{group closeness centrality} measure~\cite{everett1999centrality} (see Eq.~\ref{eq:closeness-group}) generalizes this score to vertex sets $S \subseteq V$. \\ 
    \begin{minipage}[t]{0.48\linewidth}
        \refstepcounter{equation}\label{eq:closeness-vertex}
        \[
            (\theequation) \quad c(v) \coloneq \frac{|V|-1}{\sum_{u \in V} \dist(u,v)}
        \]
    \end{minipage}\hfill
    \begin{minipage}[t]{0.48\linewidth}
        \refstepcounter{equation}\label{eq:closeness-group}
        \[
            (\theequation) \quad c(S) \coloneq \frac{|V|-|S|}{\sum_{u \in V} \dist(u,S)}
        \]
    \end{minipage}
    \par\bigskip

    Note that other definitions instead use $n := |V|$ or $1$ as the numerator.
    Regardless of the exact definition of group closeness centrality, due to the identical denominator one may equivalently minimize \emph{group farness} $ f(S) \coloneq \sum_{u \in V} \dist(u, S).$

    Note that in the literature, there are varying names used to refer to this problem.
    Some authors refer to the problem itself as \emph{group closeness centrality}~\cite{staus2023exact}, while others refer to it as \emph{group closeness maximization}~\cite{DBLP:journals/corr/abs-1710-01144}.
    We refer to the problem as \emph{group closeness centrality maximization} (\gccm), to differentiate the problem from the function.

    Let $\mathcal{U}$ be a universe of items and $f : 2^{\mathcal{U}} \rightarrow \mathbb{R}$ be a set function.
    The \emph{marginal gain} of $e \in \mathcal{U}$ at $A \subseteq \mathcal{U}$ is $\Delta(e \mid A) = f(A \cup \{e\}) - f(A)$.
    A set function is \emph{submodular} if $\Delta(e \mid A) \geq \Delta(e \mid B)$ and \emph{supermodular} if $\Delta(e \mid A) \leq \Delta(e \mid B)$ for all $A \subseteq B \subseteq \mathcal{U}$, $e \in \mathcal{U} \setminus B$.
    The group farness function is supermodular~\cite{DBLP:journals/corr/abs-1710-01144}.

    \section{Related Work}\label{sec:related-work}

    \subsection{Exact Algorithms}\label{subsec:exact-algorithms}
    To exactly solve \gccm, recent research has focused on branch-and-bound algorithms~\hbox{\cite{staus2023exact, Woydt24}} or ILP solvers~\hbox{\cite{DBLP:journals/corr/abs-1710-01144, staus2023exact}}.
    Branch-and-bound algorithms explore the space of possible solutions $S \subseteq V$ with $|S| = k$ by traversing a Set-Enumeration Tree~\cite{Rymon92}.
    Let $T$ be a node of the search tree with a \emph{working set} $S_T = \{s_1, \ldots, s_i\}$, i.e. the current subset of the solution, and a \emph{candidate set} $C_T = \{c_1, \ldots, c_m\}$, i.e. the elements which can be added to $S_T$.
    A node has $m$ children where the $i$-th child has the candidate set $C_T \setminus \{c_1, \ldots, c_{i}\}$ and the working set $S_T \cup \{c_i\}$.
    A node $T$ with a working set of size $k$ has no children at all and the root $R$ of the tree has the candidate set $C_R = V$ and working set $S_R = \emptyset$.
    The leafs of this tree hold all sets $S \subseteq V$ of size $k$ exactly once, and therefore exploring the tree guarantees to find the optimal set $S^*$ that maximizes \gccm.
    The search tree has $\mathcal{O}{ n \choose k }$ nodes, however Staus~et~al.~\cite{staus2023exact} and Woydt~et~al.~\cite{Woydt24} present methods to cut off large parts of the search tree.
    They exploit the supermodularity of group farness by computing a lower bound on the achievable score in a subtree.
    If that bound is greater than the score of the best found solution then the subtree does not hold a better solution and can be cut off.
    Let $S'$ be the best set the branch-and-bound algorithm has found so far and let $T$ be the node the algorithm is currently exploring.
    The core idea is to calculate a lower bound $b \leq f(S_T \cup C_T^*)$, where $C_T^* = \arg\min_{C' \subseteq C_T, |C'| = k - |S_T|} f(S_T \cup C')$.
    This underestimates the best possible farness in the subtree of the current node $T$.
    If $f(S') \leq b$, i.e., the underestimated bound is still greater than the best solution, then no better solution exists in the subtree of $T$, and it does not have to be explored further.
    Additionally, instead of pruning complete nodes of the search tree, they also propose a method to prune individual vertices of the candidate set.
    This reduces the number of children, and therefore the size of the search tree to be explored.
    Both of these methods are based on the supermodularity of group farness.

    Bergamini~et~al.~\cite{DBLP:journals/corr/abs-1710-01144} were the first to propose an ILP formulation for \gccm.
    For every vertex $v \in V$, a variable $y_v \in \{0, 1\}$ indicates whether $v \in S^*$ ($y_v = 1$) or $v \notin S^*$ ($y_v = 0$).
    Furthermore, they define the variables $x_{u,v} \in \{0, 1\}$ for each vertex pair $(u, v) \in V \times V$.
    They say a vertex $u$ is \emph{assigned} to a vertex $v \in S^*$ if $\dist(u, v) = \dist(u, S^*)$.
    If $v \in S^*$ and $u$ is assigned to $v$ then $x_{u,v} = 1$, otherwise $x_{u,v} = 0$.
    The ILP is then written as
    \begin{subequations}\label{ilp:gccm}
        \begin{align}
            & \sum_{v\in V} y_v = k \label{simple_ilp:sum_k} \\[-2pt]
            \text{minimize} \quad \sum_{\substack{u \in V \\ v \in V}} \dist(u,v)\,x_{u,v} \quad \text{subject to} \quad & \sum_{v\in V} x_{u,v} = 1 \quad \forall\ u\in V \label{simple_ilp:u_assigned} \\[-2pt]
            & x_{u,v} \le y_v \quad \forall\ u,v\in V \label{simple_ilp:valid_assigned}
        \end{align}
    \end{subequations}

    Constraint~\ref{simple_ilp:sum_k} ensures that exactly $k$ vertices are in the solution.
    Constraint~\ref{simple_ilp:u_assigned} ensures that each vertex $u$ is assigned to exactly one vertex, while Constraint~\ref{simple_ilp:valid_assigned} ensures that $u$ can only be assigned to a vertex $v \in S^*$.
    The objective function only adds the distance $\dist(u, v)$ if $u$ is assigned to $v$.
    Since each vertex is assigned to exactly one vertex in the solution, minimizing the term minimizes group farness.
    The ILP has $\mathcal{O}(n^2)$ variables and constraints.

    Staus~et~al.~\cite{staus2023exact} proposed further ILP formulations that reduce the number of needed variables and constraints.
    The idea for all of their new formulations is to circumvent the creation of the variables $x_{u,v}$.
    After all, the information to which vertex $v$ the vertex $u$ is assigned is not inherently relevant, only the underlying distance to $S^*$ is.
    Therefore, the variables $x_{u,v}$ are removed from the ILP and instead variables $x_{u,i} \in \{0, 1\}$ are added with $u \in V$ and $i \in \{0, \ldots, \diam(G)\}$.
    If vertex $u$ has distance $i$ to the optimal solution $S^*$, then $x_{u,i} = 1$, otherwise $x_{u,i} = 0$.
    The new ILP formulation then reads as follows:

    \begin{subequations}\label{ilp:staus}
        \begin{align}
            &  \sum_{v \in V} x_{v,0} = k  \label{smaller_ilp:sum_k} \\[-2pt]
            \text{minimize}  \sum_{\substack{u \in V \\ i \in \{0,\ldots,\diam(G)\}}} i \cdot x_{u,i} \quad \text{subject to} \quad\quad\quad & \hspace{-0.8cm}\sum_{i \in \{0, \ldots, \diam(G)\}} x_{v,i} = 1 \quad \forall\, v \in V     \label{smaller_ilp:one_distance} \\[-2pt]
            & \hspace{-1.38cm}x_{v,i} \le \sum_{\substack{w \in V \\ \dist(v,w) = i}} x_{w, 0} \quad \forall\, \substack{v \in V \\ i \in \{0,\ldots,\diam(G)\}}  \label{smaller_ilp:valid_distance}
        \end{align}
    \end{subequations}

    Constraints~\ref{smaller_ilp:sum_k} and~\ref{smaller_ilp:one_distance} encode that the solution has exactly $k$ vertices and each vertex has only one active distance to $S^*$.
    Constraint~\ref{smaller_ilp:valid_distance} ensures that the active distance from $v$ to $S^*$ is chosen correctly, i.e., $x_{v,i}$ can only be 1 if there exists at least one $w \in S^*$ with distance $i = \dist(v, w)$.
    After optimization, $S^* = \{v \mid x_{v, 0} = 1\}$ is a globally optimal solution.
    This ILP has $\mathcal{O}(n \cdot \diam(G))$ variables and constraints, making it inherently well suited for small diameter networks like social networks.
    The ILP can further be reduced in size by removing unnecessary variables.
    Instead of using $i \in \{0, \ldots, \diam(G)\}$, one can instead use $i \in \{0, \ldots, \ecc(v)\}$ for each individual vertex $v$ since $\dist(v, S^*) \leq \ecc(v)$.

    However, Staus~et~al.~\cite{staus2023exact} still noticed that many $x_{v,i}$ are not relevant for the ILP\@.
    For example, if a solution with $x_{v,i} = 1$ was determined, then all variables $x_{v,j}$ with $j > i$ do not impact the objective function.
    An ILP without those variables will determine the same solution in less time.
    This led them to design an iterative algorithm called \textsc{ILPind}~\cite{staus2023exact}, which is the state-of-the-art algorithm to exactly solve \gccm.
    The idea of the algorithm is to start with a small formulation of the ILP which might not be able to find the optimal solution - in which case the ILP is iteratively enlarged until this is the case.
    Instead of creating all variables $x_{v, i}$ with $i \in \{0, \ldots, \ecc(v)\}$, they only create the variables with $i \in \{0, \ldots, d(v) = 2\}$.
    The variables $x_{v, i}$ with $i < d(v)$ keep their original meaning, however the variables $x_{v, d(v)}$ change their meaning.
    If $x_{v, d(v)} = 1$ then $\dist(v, S^*) \geq d(v)$, i.e., the distance from $v$ to $S^*$ may be greater than $d(v)$, and therefore it is not guaranteed that $S^*$ is optimal.
    If the ILP is solved and there is a $x_{v, d(v)} = 1$ with $d(v) < \ecc(v)$, then the ILP is called \emph{insufficient} and another iteration is needed\footnote{Note that this notion of sufficiency is relative to the solution found. If several optimal solutions exist, the ILP might be sufficient with respect to one solution, but insufficient with respect to another.}.
    For each vertex $v$ with $x_{v, d(v)} = 1$ the value $d(v)$ is increased by 1 but capped at $\ecc(v)$.
    The ILP is solved with the new $d(v)$-values and this is repeated until eventually a sufficient ILP is encountered.
    While not obviously advantageous over solving the original ILP once, experiments in~\cite{staus2023exact} show it is significantly faster.

    Another improvement they utilize is the use of dominating vertices.
    Remember that a vertex $v$ is dominated by a vertex $u$ if $N[v] \subseteq N[u]$.
    Staus et al. \cite{staus2023exact} show that a set of dominated vertices $D$ can be removed from the search space, i.e., there exists an optimal solution $S^*$ in $V \setminus D$.
    More precisely, for all vertices $v \in V$ it must hold that $v \in V \setminus D$ or $v$ is dominated by a vertex in $V \setminus D$.
    Let $D \subseteq V$ be the set of all vertices removed from the search space according to this rule.
    Then, for all vertices $v \in D$, the ILP does not need the variables $x_{v, 0}$.
    Note that $V \setminus D$ is constructed to have at least size $k$.

    \subsection{Heuristic and Approximate Algorithms}\label{subsec:heuristic-algorithms}
    Due to the NP-hardness of \gccm, many authors have proposed heuristic algorithms.
    Chen~et~al.~\cite{chen2016efficient} proposed a simple greedy heuristic, which we call \greedy.
    Essentially, \greedy\ starts with the empty set $S_0 = \emptyset$ and iteratively increases $S_{i+1} = S_i \cup \{v\}$ with the vertex $v$ that yields the greatest marginal gain. %
    Then, the computed solution is $S_k$.
    \greedy\ first computes the $\mathcal{O}(n^2)$ size distance matrix in $\mathcal{O}(n(n+m))$ time by BFS from each vertex.
    To determine $v$ takes $\mathcal{O}(n^2)$ time in each of the $k$ iterations.
    Chen~et~al.~\cite{chen2016efficient} claimed that \greedy\ has an approximation ratio of $(1-1/e)$ by using a well-known result for submodular set functions~\cite{nemhauser1978analysis}.
    Bergamini~et~al.~\cite{DBLP:journals/corr/abs-1710-01144} however pointed out that group closeness centrality is not submodular and thus showed that this approximation proof strategy does not work.
    Consequently, the question of whether \gccm\ could be approximated was considered unsettled~\cite{angriman2021group}.
    We later show that \greedy\ does not have an approximation guarantee.
    However, Bergamini~et~al.~\cite{DBLP:journals/corr/abs-1710-01144} observed very high empirical solution quality and proposed \greedy++, an improved algorithm that computes the same solution as \greedy\ while using less time and space.
    They only maintain a vector storing the distance from each vertex to the current set $S_i$ and use $n$ BFS queries per iteration, many of which can be terminated early or skipped entirely.

    Angriman~et~al.~\cite{DBLP:journals/corr/abs-1911-03360} proposed two local search algorithms, \textsc{LocalSwap} and \textsc{GrowShrink}, with no known approximation guarantee.
    Both start by uniformly picking an initial solution~$S$.
    \textsc{LocalSwap} operates by swapping vertices $v \in S$ with their neighbors $u \in N(v)$ if it is expected that the group farness decreases.
    Instead of querying $f(\{u\} \cup S \setminus \{v\})$ for each possible swap, they heuristically determine good pairs of vertices, and keep auxiliary data structures to speed up the query of $f$.
    If no swap is found, then the algorithm terminates and returns the current set.
    Their second algorithm, \textsc{GrowShrink}, operates in two distinct phases, the \emph{grow}- and \emph{shrink}-phase.
    During the grow-phase the algorithm adds additional vertices and temporarily violates the constraint $|S| = k$.
    In the shrink-phase the algorithm discards vertices such that the constraint is again fulfilled.
    The idea is that the grow-phase adds beneficial vertices, while the shrink-phase then discards the least beneficial ones.

    In a later work, Angriman~et~al.~\cite{angriman2021group} propose an algorithm with a \hbox{1/5-approximation} guarantee, thus settling the approximability question of \gccm.
    The local search algorithm is based on an algorithm for the related $k$-median problem analyzed in Arya~et~al.~\cite{arya2001local}.
    They start with any solution $S$ and then perform swaps $(S \setminus \{u\}) \cup \{v\}$ with $u \in S$ and $v \in V \setminus S$ if the swap is beneficial.
    Once all possible swaps do not improve the objective anymore, the algorithm terminates.
    They show that the result of Arya et al.~\cite{arya2001local} applies, i.e. for the output set $S$ it holds that $c(S) \geq \frac15c(S^*)$ (i.e. $f(S) \leq 5 f(S^*)$) with $S^*$ being an optimal solution.
    Since the approximation guarantee does not depend on the initial solution, different algorithms to obtain the initial solutions can be chosen.
    Angriman~et~al. propose two algorithms: \textsc{Greedy-Ls-C} and \textsc{GS-LS-C}, where the first one uses the solution provided by the \greedy\ algorithm and the second one utilizes the \textsc{GrowShrink} algorithm.
    Interestingly, \textsc{Greedy-Ls-C}'s and \textsc{GS-LS-C}'s solution quality are very similar.

    Finally, Rajbhandari~et~al.~\cite{rajbhandari2022presto} proposed an anytime algorithm called \textsc{Presto} that finds heuristic solutions and, when run until termination, exact solutions. %

    \section{Grover} \label{sec:grover}
    The main contribution of this work is our exact algorithm \textsc{Grover}\footnote{\textbf{Gro}up Closeness Centrality Maximization Sol\textbf{ver}}.
    It addresses two structural limitations of \textsc{ILPind}~\cite{staus2023exact}: the ILP contains many variables whose values are fully determined by other variables, and the iterative framework wastes many iterations before the ILP becomes sufficient.
    We resolve the first limitation through absorbed vertices, a new reduction that eliminates variables by folding their cost into their neighbor.
    The second is resolved by bootstrapping the ILP with an approximate solution, producing near-sufficient formulations from the first iteration.

    \textbf{Recap.}
    Recall that the algorithm \textsc{ILPind} from Staus~et~al.~\cite{staus2023exact} starts by solving the ILP with variables $x_{v, i}$ for $v \in V$ and $i \in \{0, \ldots, d(v) = 2\}$.
    This way, solutions with $\dist(v,S^*) \leq 1$ or $\dist(v,S^*) = 2$ and $\ecc(v) = 2$ can be represented by the ILP\@.
    For instance, if a dominating set of size $k$ exists, then this set is a solution which can be represented by the initial ILP\@.
    After the model is optimized with the current $d(v)$-values, we may find that $x_{v,d(v)} = 1$ for some vertices $v$, indicating that $\dist(v,S^*) \geq d(v)$.
    If, in addition, $\ecc(v) > d(v)$, then we cannot rule out the possibility that $\dist(v,S^*) > d(v)$.
    For such vertices $v$, the value $d(v)$ is incremented by 1 to make sure that a possible solution $S^*$ with $\dist(v,S^*) > d(v)$ can be represented by the ILP model.
    This process repeats until eventually a sufficient ILP is obtained.

    \subsection{Reducing Iterations}\label{subsec:reducing-iterations}
    The iterative framework of \textsc{ILPind} faces a fundamental problem.
    Setting $d(v) = 2$ produces a small ILP, but on graphs with a moderate to large diameter, the first ILPs cannot represent any optimal solution and many iterations are wasted incrementally enlarging the model.
    Conversely, setting $d(v) = \ecc(v)$ immediately yields a sufficient ILP, but it has many unnecessary variables that slow down the solver.
    Neither extreme is satisfactory: the former pays in iteration count, the latter in per-iteration cost.

    We resolve this problem by using an estimate $2 \leq \widetilde{d}(v) \leq \ecc(v)$ for each vertex $v$.
    The closer these estimates are to the final $d(v)$-values needed to represent an optimal solution, the fewer iterations are needed.
    If indeed $\widetilde{d}(v) = \dist(v, S^*) + 1$ and a unique optimal solution $S^*$ exists, the ILP is sufficient in the first iteration with as few variables as possible.
    Note that if multiple optimal solutions exist, then additional iterations might be necessary to verify that no other optimal solution is strictly better.
    The key question is how to obtain good estimates $\widetilde{d}(v)$ without knowing $S^*$.
    Heuristic algorithms for \gccm\ produce solutions $\widetilde S$ empirically very close to optimal on real-world graphs~\cite{angriman2021group, angriman2023algorithms, DBLP:journals/corr/abs-1710-01144}, so $\dist(v, \widetilde{S})$ closely tracks the unknown $\dist(v, S^*)$ and can be used to determine $\widetilde{d}(v)$.
    Concretely, we derive $\widetilde{S}$ using \textsc{GS-LS-C} from Angriman~et~al.~\cite{angriman2023algorithms} due to its efficiency and 5-approximation and set $\widetilde{d}(v) = \max\{\dist(v, \widetilde{S}) + 1, 2\}$ for the first iteration.
    While we use \textsc{GS-LS-C} for its efficiency and approximation guarantee, any algorithm producing an initial solution $\widetilde{S}$ can be used, since the iterative enlargement of the ILP ensures that optimality is guaranteed regardless of the quality of $\widetilde{S}$.

    \subsection{Absorbed Vertices}\label{subsec:absorbed-vertices}
    Staus~et~al.~\cite{staus2023exact} showed that a set $D \subseteq V$ of dominated vertices can be removed from the search space without compromising the optimality guarantee.
    This is expressed by removing the variables $x_{v, 0}$ for each $v \in D$ from the ILP\@.
    However, this reduction only affects the search space, as distance variables $x_{v, 1}, \ldots, x_{v, d(v)}$ and their constraints remain, since they are needed to compute the objective function, and the size of the ILP is largely unchanged.

    Absorbed vertices go fundamentally further by eliminating \emph{all} variables of a vertex from the ILP\@.
    The key insight is that certain vertices have their distance to any optimal solution fully determined by a neighbor's distance.
    This allows us to fold their associated costs into the neighbor and completely remove them from the ILP\@.
    For example, consider a vertex $v$ with $\deg(v) = 1$ and let $u$ be its neighbor.
    Since all shortest paths from $v$ to $S^*$ must pass through the edge $\{u, v\}$, we know $\dist(v, S^*) = \dist(u, S^*) + 1$.
    We therefore can deduce $x_{u, i} = x_{v, i+1}$ in any solved ILP\@.
    Since we are able to deduce the assignment of $x_{v, i+1}$ from $x_{u,i}$ we can remove it from the ILP\@.
    However, removing the variable alters the objective function of the ILP\@.
    Therefore, the cost associated with $x_{v, i+1}$ must be folded into the cost of $x_{u, i}$, such that the same function is still computed.
    The variable $x_{u, i}$ also has to capture the distance to $v$, so instead of associating a cost of $i$ with it, a cost of $2i + 1$ is now associated with the variable.
    In this case, we say that a vertex $u$ \emph{absorbs} a vertex $v$ and, in general, a vertex $u$ may absorb $\alpha(u) \in \mathbb{N}$ many neighbors.
    In that case, $u$ can fold the cost of these vertices into it by associating the cost $\alpha(u)(i + 1) + i$ with $x_{u, i}$ in the objective function.
    The set $A \subseteq V$ denotes all absorbed vertices in $G$.

    There are even more cases when a vertex $u$ can absorb its neighbors.
    Let $u$ be a cut vertex, i.e., $G - u$ has multiple connected components.
    If there is a component $C = \{w_1, \ldots, w_\ell\} \subseteq D$ in $G-u$ such that all its vertices are dominated by $u$, then $u$ can absorb $w_1, \ldots, w_\ell$.
    In that case, each shortest path from $w_1, \ldots, w_\ell$ to $S^*$ must pass through $u$ and the same argument from above applies.
    Figure~\ref{fig:absorbed-vertices} shows multiple cases where vertices can and cannot be absorbed.
    A proof that this search space reduction is correct is given in the next section.

    \begin{figure}
        \centering
        \begin{minipage}[t]{0.32\textwidth}
            \centering
            \begin{tikzpicture}[
    vertex/.style={
        circle,
        draw,
        fill=#1,
        inner sep=0pt,
        minimum size=3.5mm
    }
]
    \def\colorb{blue!60}
    \def\colorr{red!60}
    \def\colorblack{black!60}

    \def\vdist{0.8}

    \node[vertex=\colorr]      (v0) at (0,  \vdist) {};
    \node[vertex=\colorb]      (v1) at (0,  0) {};
    \node[vertex=\colorblack]  (v2) at (0, -\vdist) {};
    \node[vertex=\colorb]  (v3) at (0, -2*\vdist) {};
    \node[vertex=\colorr]  (v4) at (-0.5*\vdist, -3*\vdist) {};
    \node[vertex=\colorr]  (v5) at (0.5*\vdist, -3*\vdist) {};

    \draw (v0) -- (v1);
    \draw (v1) -- (v2);
    \draw (v2) -- (v3);
    \draw (v3) -- (v4);
    \draw (v3) -- (v5);
    \draw (v4) -- (v5);
\end{tikzpicture}
            \par\smallskip
            (a)
        \end{minipage}\hfill
        \begin{minipage}[t]{0.32\textwidth}
            \centering
            \begin{tikzpicture}[
    vertex/.style={
        circle,
        draw,
        fill=#1,
        inner sep=0pt,
        minimum size=3.5mm
    }
]
    \def\colorb{blue!60}
    \def\colorr{red!60}
    \def\colorblack{black!60}

    \def\vdist{0.8}
    \def\rad{1.0} %

    \node(x) at (0,  0) {};
    \node[vertex=\colorblack]      (v0) at (0,  \vdist) {};
    \node[vertex=\colorb]      (v1) at (0,  3*\vdist) {};
    \node     (y1) at (-0.5*\vdist,  0) {...};
    \node     (y2) at (0.5*\vdist,  0) {...};

    \node[vertex=\colorr] (v2) at ($(v1) + (180:\rad)$) {};
    \node[vertex=\colorr] (v3) at ($(v1) + (120:\rad)$) {};
    \node[vertex=\colorr] (v4) at ($(v1) + ( 60:\rad)$) {};
    \node[vertex=\colorr] (v5) at ($(v1) + (  0:\rad)$) {};
    \node[vertex=\colorr] (v6) at ($(v1) + (  -60:\rad)$) {};
    \node[vertex=\colorr] (v7) at ($(v1) + (  -120:\rad)$) {};

    \draw (v0) -- (v1);
    \draw (v0) -- (y1);
    \draw (v0) -- (y2);

    \draw (v1) -- (v2);
    \draw (v1) -- (v3);
    \draw (v1) -- (v4);
    \draw (v1) -- (v5);
    \draw (v1) -- (v6);
    \draw (v1) -- (v7);

    \tikzset{spaced dashed/.style={dash pattern=on 3pt off 8pt}}

    \draw[spaced dashed] (v2) -- (v3);
    \draw[spaced dashed] (v2) -- (v4);
    \draw[spaced dashed] (v2) -- (v5);
    \draw[spaced dashed] (v2) -- (v6);
    \draw[spaced dashed] (v2) -- (v7);

    \draw[spaced dashed] (v3) -- (v4);
    \draw[spaced dashed] (v3) -- (v5);
    \draw[spaced dashed] (v3) -- (v6);
    \draw[spaced dashed] (v3) -- (v7);

    \draw[spaced dashed] (v4) -- (v5);
    \draw[spaced dashed] (v4) -- (v6);
    \draw[spaced dashed] (v4) -- (v7);

    \draw[spaced dashed] (v5) -- (v6);
    \draw[spaced dashed] (v5) -- (v7);

    \draw[spaced dashed] (v6) -- (v7);
    \node[vertex=\colorb] at (v1) {};
\end{tikzpicture}
            \par\smallskip
            (b)
        \end{minipage}\hfill
        \begin{minipage}[t]{0.32\textwidth}
            \centering
            \begin{tikzpicture}[
    vertex/.style={
        circle,
        draw,
        fill=#1,
        inner sep=0pt,
        minimum size=3.5mm
    }
]
    \def\colorb{blue!60}
    \def\colorr{red!60}
    \def\colorblack{black!60}

    \def\vdist{0.8}

    \node[vertex=\colorblack]      (v0) at (0,  \vdist) {};
    \node[vertex=\colorblack]      (v1a) at (-0.5*\vdist,  0) {};
    \node[vertex=\colorblack]      (v1b) at (0.5*\vdist,  0) {};
    \node[vertex=\colorblack]  (v2) at (0, -\vdist) {};
    \node[vertex=\colorblack]  (v3a) at (-0.5*\vdist, -2*\vdist) {};
    \node[vertex=\colorblack]  (v3b) at (0.5*\vdist, -2*\vdist) {};
    \node[vertex=\colorblack]  (v4) at (-0.5*\vdist, -3*\vdist) {};
    \node[vertex=\colorblack]  (v5) at (0.5*\vdist, -3*\vdist) {};

    \draw (v0) -- (v1a);
    \draw (v0) -- (v1b);
    \draw (v1a) -- (v2);
    \draw (v1b) -- (v2);
    \draw (v2) -- (v3a);
    \draw (v2) -- (v3b);
    \draw (v3a) -- (v3b);
    \draw (v3a) -- (v4);
    \draw (v3b) -- (v5);
    \draw (v4) -- (v5);
\end{tikzpicture}
            \par\smallskip
            (c)
        \end{minipage}

        \caption{Figure showing absorbed vertices. In case (a) and (b) the red vertices can be absorbed by their neighboring blue vertex. In case (b) any combination of the dashed edges results in the blue vertex absorbing its neighbors. In case (c) no vertex can be absorbed.}
        \label{fig:absorbed-vertices}
    \end{figure}

    \textbf{Runtime Analysis.}
    The dominated vertices can be determined in $\mathcal{O}(m\Delta(G))$ time by sorting each vertex's neighborhood and performing pairwise linear scans.
    To compute the set of absorbed vertices $A$, Tarjan's~\cite{Tarjan72} $\mathcal{O}(n + m)$ biconnectivity algorithm identifies cut vertices, and for each cut vertex we check if it dominates a complete component.
    Since the vertex has at most $\Delta(G)$ neighbors and we already computed the domination relation between all neighboring vertices, this costs at most $\mathcal{O}(\Delta(G))$ time per vertex.
    For at most $n$ cut-vertices this takes at most $\mathcal{O}(n \Delta(G))$ time.
    In total, a runtime of $\mathcal{O}(m\Delta(G))$ is required.

    \subsection{The Complete Algorithm}\label{subsec:the-complete-algorithm}
    \textsc{Grover} starts by determining the set of dominated vertices $D$ and the set of absorbed vertices $A \subseteq D$.
    Afterwards, \textsc{GS-LS-C}~\cite{angriman2023algorithms} is used to compute an initial solution $\widetilde{S}$.
    The values $\widetilde{d}(v) = \max\{\dist(v, \widetilde{S}) + 1, 2\}$ are determined for each vertex $v$ for the first iteration of our ILP\@.
    The ILP then reads as follows:

    \begin{subequations}\label{ilp:reduced}
        \begin{align}
            & \sum_{v \in V \setminus D} x_{v,0} = k \label{reduced_ilp:sum_k} \\[-2pt]
            \min\ \sum_{\substack{v \in V \setminus A \\ i \in \{0,\ldots,\widetilde{d}(v)\}}} x_{v,i}\,\bigl(\alpha(v)\cdot(i+1)+i\bigr) \quad \text{s. t.} \quad & \hspace*{-0.33cm} \sum_{i \in \{0,\ldots,\widetilde{d}(v)\}} x_{v,i} = 1 \quad \forall\, v \in V \setminus A \label{reduced_ilp:one_distance} \\[-2pt]
            & \hspace*{-0.25cm} \sum_{\substack{w \in V \setminus D \\ \dist(v,w)=i}} x_{w,0} \geq x_{v,i} \quad \forall\, \substack{v \in V \setminus A \\ i \in \{0,\ldots,\widetilde{d}(v)\}}
            \label{reduced_ilp:valid_distance}
        \end{align}
    \end{subequations}

    Note that in the worst case, the reduced ILP still has the same number of variables and constraints as \textsc{ILPind}, namely $\mathcal{O}(n \cdot \diam(G))$.
    If, after solving the ILP, we encounter $x_{v, \widetilde{d}(v)} = 1\ \land\ \ecc(v) > \widetilde{d}(v)$ for any $v \in V \setminus A$, then the ILP may be insufficient to represent an optimal solution.
    In this case, we increase $\smash{\widetilde{d}(v) \gets \widetilde{d}(v) + 1}$ for all vertices that did not pass the check and solve the ILP again with the new $\widetilde{d}(v)$-values.
    This is repeated until the ILP is sufficient and then the optimal solution $S^* = \{v \in V\setminus D \mid x_{v, 0} = 1\}$ can be extracted.

    For high-level pseudocode of the full algorithm, see Algorithm~\ref{alg:grover}.
    Lines $2$ and $3$ handle the case for $k=1$, where the \greedy\ algorithm returns the exact solution.
    In Line 4, the dominated and absorbed vertex sets are determined and in Line 5, the heuristic solution is computed using \textsc{GS-LS-C}~\cite{angriman2023algorithms}.
    This is used to estimate $\widetilde{d}(v)$ for each vertex $v$ in Lines $6$ and $7$.
    The main work of the algorithm occurs in the loop from Lines $8$ to $16$.
    In Line $9$ the reduced ILP is solved, and Lines $10$ to $14$ check if the ILP is sufficient.
    for each vertex $v \in V \setminus A$ it is checked whether the ILP is insufficient and if so the corresponding $\widetilde{d}(v)$-value is incremented (Line 14).
    Only when all vertices passed the check, will the algorithm return the solution in Line $16$.

    \begin{algorithm}[H]
    \caption{\textsc{Grover}}\label{Grover}
    \begin{algorithmic}[1]
            \Function{solve($G$: Graph, $k: \mathbb{N}$)}{}
                \If{$k = 1$}
                    \State \textbf{return} solution of \textsc{Greedy} with $k=1$ \Comment{\textsc{Greedy} is exact for $k = 1$}
                \EndIf
                \State compute $D, A$ and $\{\alpha(v)\}_{v \in V\setminus A}$\label{lst:line:comp_A}
                \State compute approximate solution $\widetilde S$ \Comment{Using  \textsc{GS-LS-C}}
                \For {$v \in V\setminus A$}
                    \State $\widetilde d(v) \gets \max\{2,\dist(v, \widetilde S) + 1\}$ \Comment{See Section \ref{subsec:reducing-iterations}}
                \EndFor
                \While {\emph{true}}
                    \State
                    solve reduced ILP using $D,A, \{\widetilde d(v)\}_{v \in V\setminus A}$ and $\{\alpha(v)\}_{v \in V\setminus A}$ \Comment{See Section \ref{subsec:absorbed-vertices}}
                    \State ilpSufficient $ \gets $ \emph{true}
                    \For {$v \in V\setminus A$}
                        \Comment{Check if the ILP is sufficient}
                        \If{$x_{v,\widetilde d(v)} = 1$ and $\widetilde d(v) < \ecc(v)$}
                            \Comment{ is $\widetilde d(v)$ insufficient for $v$?}
                            \State ilpSufficient $ \gets $ \emph{false}
                            \State $\widetilde d(v) \gets \widetilde d(v) + 1$
                        \EndIf
                    \EndFor
                    \If{ilpSufficient}
                        \State \textbf{return} $\{v \in V\setminus D \mid x_{v,0} = 1\}$
                    \EndIf
                \EndWhile
            \EndFunction
    \end{algorithmic}
    \label{alg:grover}
    \end{algorithm}

    \textbf{Proof of Correctness.}
    To prove the correctness of our algorithm, we show that the new ILP formulation indeed finds an optimal set $S^*$.
    We accomplish this by showing that the new ILP formulation and the \textsc{ILPind} formulation (see Section~\ref{subsec:exact-algorithms}) optimize the same function.
    Additionally, we show that we can infer the variable assignment of one ILP based on the variable assignment of the other.
    Since \textsc{ILPind} computes an optimal set $S^*$ (see~\cite{staus2023exact} for a proof), so will our algorithm.
    Note that reducing the number of iterations by estimating the initial $\widetilde d(v)$-values does not impact the correctness guarantee.

    \begin{theorem}
        Let a reduced optimized ILP be sufficient, i.e., after optimizing the ILP it holds that $x_{v,\widetilde d(v)} = 0$ with $\widetilde d(v) < \ecc(v)$ or $x_{v, \widetilde d(v)} = 1$ with $\widetilde d(v) \geq \ecc(v)$.
        Then $S^* = \{v \in V \setminus D \mid x_{v, 0} = 1\}$ is a globally optimal solution.
    \end{theorem}

    \begin{proof}
        Let $\{\widetilde x_{v,i} \mid 0 \leq i \leq \widetilde d(v)\}_{v \in V\setminus A}$ be the variables of the reduced model with values $\{\widetilde{d}(v)\}_{v \in V \setminus A}$ of the last iteration.
        By the definition of our algorithm, these values correspond to an optimum solution of the ILP and are sufficient, i.e., for each $v \in V\setminus A$, we have $\widetilde x_{v,\widetilde d(v)} = 0$ or both $\widetilde x_{v,\widetilde d(v)} = 1$ and $\ecc(v) \leq \widetilde d(v)$.

        Consider the original ILP model used in \textsc{ILPind}, for which we introduce variables \hbox{\smash{$\{ x_{v,i} \mid 0 \leq i \leq d(v)\}_{v \in V}$}}.
        Let $d(v) = \widetilde{d}(v)$ if $v \in V \setminus A$ and $d(v) = \widetilde{d}(\rho(v))+1$ if $v \in A$ where $\rho(v)$ denotes the vertex which absorbed $v$.
        We assign
        \begin{align}
            \label{var_transform}
            x_{v,i} = \begin{cases}
                          0 & v \in A, i = 0 \\
                          \widetilde x_{ \rho(v), i-1} & v \in A, i \geq 1 \\
                          \widetilde x_{v,i} & v \in V \setminus A
            \end{cases}  \quad\quad \quad \text{for }  0 \leq i \leq d(v).
        \end{align}
        This construction yields a feasible solution of the original ILP model (which is also sufficient by construction).
        It remains to show that this variable assignment corresponds to an optimum solution.
        First note that by the reduction rules, the identity
        \begin{align}
            \label{redu_identity}
            \sum_{v \in V\setminus A} \ \sum_{ i \in\{0,\dots,\widetilde d(v)\}} \alpha(v)\cdot (i+1) \cdot \widetilde x_{v,i} = \sum_{v \in A} \ \sum_{ i \in\{0,\dots,\widetilde d( \rho(v))\}}  (i+1) \cdot \widetilde x_{ \rho(v),i}
        \end{align}
        holds.
        Roughly speaking, this identity expresses that the absorbed cost can be counted either by considering the absorbing vertices, or by considering the absorbed vertices.
        For the objective values $\widetilde f_*$ of the solution of the reduced ILP, and $f_*$ of the ILP used in \textsc{ILPind}, we then get
        \begin{equation}
            \begin{aligned}
                \widetilde f_* &= \sum_{v \in V\setminus A} \ \sum_{ i \in\{0,\dots,\widetilde d(v)\}} (\alpha(v)\cdot (i+1)+i) \cdot \widetilde x_{v,i}
                \\
                &= \sum_{v \in V\setminus A} \ \sum_{ i \in\{0,\dots,\widetilde d(v)\}} i \cdot \widetilde x_{v,i} + \sum_{v \in V\setminus A} \ \sum_{ i \in\{0,\dots,\widetilde d(v)\}} \alpha(v)\cdot (i+1) \cdot \widetilde x_{v,i}
                \\
                &\overset{\eqref{redu_identity}}=
                \sum_{v \in V \setminus A} \ \sum_{ i \in\{0,\dots,\widetilde d(v)\}} i \cdot \widetilde x_{v,i}+ \sum_{v \in A}\sum_{ i \in\{0,\dots,\widetilde d( \rho(v))\}}   (i+1) \cdot\widetilde x_{ \rho(v),i}
                \\
                &\overset{(\star)}=\sum_{v \in V \setminus A} \ \sum_{ i \in\{0,\dots,d(v)\}} i \cdot x_{v,i}+ \sum_{v \in A}\sum_{ i \in\{0,\dots,d(v)\}}   i \cdot x_{v,i}
                \\
                &=\sum_{v \in V} \ \sum_{ i \in\{0,\dots,d(v)\}} i \cdot  x_{v,i} = f_*
            \end{aligned}\label{f_transform}
        \end{equation}
        where ($\star$) follows by substituting in $d(v)$ and $\eqref{var_transform}$ together with an index shift.
        Hence, the constructed solution of the original ILP used in \textsc{ILPind} has the same objective value $f_* = \widetilde f_*$ as the solution of the reduced ILP\@.
        However, we have yet to show that no strictly better solution of the original ILP exists.
        Therefore, assume for a contradiction that the constructed solution of the original model is \emph{not} optimum, i.e., there is a feasible assignment of the variables of the original model with objective value $f'<f_*$.
        We refer to the corresponding variable assignment as $\{ \hat x_{v,i} \mid 0 \leq i \leq d(v)\}_{v \in V}$.
        We argue that from this \emph{better} solution, we can construct a \emph{better} solution of the reduced ILP, contradicting its optimality.
        Indeed, using $\widetilde x_{v,i} = \hat x_{v,i}$ for all $v \in V\setminus A$ and $0 \leq i \leq d(v)$ as an assignment for the reduced ILP, we obtain a feasible solution of the reduced ILP\@.
        Applying the same rearrangements as in $\eqref{f_transform}$, we conclude that this solution has objective value $f'< \widetilde f_*$, a contradiction to the optimality of the solution of the reduced ILP\@.
    \end{proof}

    \section{Approximation Results}\label{sec:approximation-results}
    In this section, we present two results concerning the approximation of \gccm.
    The first one is the use of dominating vertices in \textsc{GS-LS-C}~\cite{angriman2023algorithms}, a $1/5$-approximation algorithm.
    The second result is a proof that the \greedy\ algorithm does not have an approximation guarantee.

    \textbf{Dominating Vertices in \textsc{GS-LS-C}.}
    Due to the NP-hardness of \gccm, heuristic and approximation algorithms are the only viable option to compute solutions on graphs with millions of vertices.
    Angriman~et~al.~\cite{angriman2021group} present an approximation algorithm that is based on a $5$-approximation algorithm~\cite{arya2001local} for the $k$-median problem.
    They show that the algorithm translates to a $5$-approximation of group farness and therefore a $1/5$-approximation for \gccm.
    The algorithm is based on objective-improving vertex swaps.
    That is, as long as two vertices $s \in S$ and $o \in V \setminus S$ exist such that $f(\{o\} \cup S \setminus \{s\}) < f(S)$, then the current solution $S$ is updated to $\{o\} \cup S \setminus \{s\}$.
    Only if no more improving swaps exist, the algorithm terminates and the aforementioned approximation can be guaranteed.

    Angriman~et~al.~\cite{angriman2021group} already restrict the search space by not considering swaps with degree-1 vertices.
    This optimization can be done since using the neighbor of a degree-1 vertex, instead of the vertex itself, never worsens the solution quality.
    The same idea can be extended to the previously introduced set of dominated vertices $D$.
    For each swap that considers a vertex $v \in D$ there is an equally good or better improving swap that utilizes a vertex $u \in V \setminus D$.
    Therefore, the approximation guarantee of Arya et al.~\cite{arya2001local} remains valid when restricting the search space to $V\setminus D$.

    \textbf{\greedy\ Does Not Approximate.}
    As mentioned in the related work, Chen~et~al.~\cite{chen2016efficient} presented the \greedy\ algorithm, which starts from the empty set $S_0 = \emptyset$ and iteratively adds the vertex that improves the score the most.
    That is, $S_i = S_{i - 1} \cup \{v\}$ with \hbox{$v = \arg\max_{v \in V \setminus S_{i-1}} c(S_{i -1 } \cup \{v\})$}.
    They mistakenly assumed that $c(\cdot)$ is submodular and that the $(1 - 1/e)$-approximation for submodular functions from Nemhauser~et~al.~\cite{nemhauser1978analysis} holds.
    However, Bergamini~et~al.~\cite{DBLP:journals/corr/abs-1710-01144} showed that $c(\cdot)$ is not submodular and therefore the approximation claim cannot be supported by this proof.
    Note that the lack of submodularity does not necessarily mean the greedy strategy performs poorly.
    Bergamini~et~al.~\cite{DBLP:journals/corr/abs-1710-01144} observed very high empirical approximation ratios never lower than $0.97$ in experiments on real-world graphs.
    While \greedy\ performs well in practice, whether it does have an approximation guarantee was left unanswered.
    Here, we show that \greedy\ can give solutions with arbitrarily bad approximation ratios, and it therefore has no constant factor approximation guarantee.

    \begin{theorem}
        \label{greedy_theorem}
        \greedy\ does not have a constant factor approximation guarantee.
    \end{theorem}

    \begin{proof}
        We show the claim by constructing a family of graphs $(G_r)_{r\in \mathbb N}$ for which \greedy\ may produce arbitrarily poor approximations for $k=2$.
        A graph $G_r$ of the family contains a path $P$ with end vertices $e_1$ and $e_2$.
        $P$ has an odd number of vertices (precisely $2r-1$) and therefore has a central vertex $c$.
        Attached to $e_1$ and $e_2$ are $r^2$ leaves (degree 1 vertices), respectively.
        Figure~\ref{fig:greedy-counter} shows an illustration of such a graph.
        Let us denote by $S^G_r$ and $S^*_r$ the solution produced by \greedy\ and the exact solution for $G_r$ with $k=2$, respectively.
        Recall that $f(S) = \sum_{v \in V(G)} \dist(v,S)$ denotes the group farness of $S$.
        Now consider the solution of \greedy.
        The algorithm first chooses the vertex $c$ as it is the most central one.
        Afterwards, the vertices $e_1$ and $e_2$ are equally well suited to be chosen and w.l.o.g. we assume the algorithm chooses $e_1$ and get $S^G_r = \{c, e_1\}$.
        The $r^2$ leaves attached to $e_2$ each contribute a cost of $r$ to $f(S^G_r)$.
        Therefore, we get $f(S^G_r) \geq r^3$.
        On the other hand, consider the solution $S_e \coloneqq \{e_1, e_2\}$.
        Clearly, $f(S^*_r) \leq f(S_e)$.
        We argue that $f(S_e)$ is upper bounded by a quadratic expression in $r$ and thus $f(S^*_r)$ is as well.
        Indeed, the leaves attached to $e_1$ and $e_2$ incur a total cost of $2r^2$ to $f(S_e)$.
        It remains to consider the cost incurred by the vertices in $P$, which is upper bounded by $2\sum_{i=1}^{r-1}i = r^2-r$.
        Therefore, overall $f(S^*_r) \leq f(S_e) \leq 3r^2-r$ and hence we get
        \begin{align*}
            \frac{f(S^G_r) }{ f(S^*_r)} \geq \frac{r^3 }{ 3r^2-r}\rightarrow \infty \text{ \ for \ } r \rightarrow \infty
        \end{align*}
        and similarly $\frac{c(S^*_r)}{c(S^G_r)} \rightarrow \infty \text{ \ for \ } r \rightarrow \infty$, showing the claim.

        \begin{figure}[t!]
            \centering
            \tikzstyle{nodeStyle} = [circle, draw, fill=black, minimum size=2mm, inner sep=0pt]
\tikzstyle{point} = [circle, draw, fill=red, minimum size=1.5mm, inner sep=0pt]
\tikzstyle{dotStyle} = [circle, draw, fill=black, minimum size=0.3mm, inner sep=0pt]
\tikzstyle{anchorStyle} = [circle, draw, fill=black, minimum size=0.0mm, inner sep=0pt]
\begin{tikzpicture}[
brc/.style = {decoration={brace,mirror,raise=0.3cm},
              decorate}]

\node[nodeStyle] (v0) at (0, -0.8) {};
\node[nodeStyle] (v1) at (0.8, -1.2) {};
\node[nodeStyle] (v2) at (1.6, -1.6) {};
\node[nodeStyle] (v3) at (2.4, -2.0) {};
\node[nodeStyle] (v4) at (3.2, -2.4) {};
\draw (v0) -- (v1);
\draw (v1) -- (v2);
\draw (v2) -- (v3);
\draw (v3) -- (v4);
\node[nodeStyle] (v5) at (4.0, -2.8) {}; %
\node[nodeStyle] (v6) at (4.8, -2.4) {};
\node[nodeStyle] (v7) at (5.6, -2.0) {};
\node[nodeStyle] (v8) at (6.4, -1.6) {};
\node[nodeStyle] (v9) at (7.2, -1.2) {};
\node[nodeStyle] (v10) at (8.0, -0.8) {};
\draw (v4) -- (v5);
\draw (v5) -- (v6);
\draw (v6) -- (v7);
\draw (v7) -- (v8);
\draw (v8) -- (v9);
\draw (v9) -- (v10);
\node[anchorStyle] (a0) at (-0.0, -0.7) {};
\node[anchorStyle] (a1) at (4.0, -2.7) {};

\node[anchorStyle] (a2) at (-1, -0.1) {};
\node[anchorStyle] (a3) at (1, -0.1) {};
\node[nodeStyle] (v11) at (-0.3, 0) {};
\node[nodeStyle] (v12) at (-0.6, 0) {};
\node[nodeStyle] (v13) at (-0.9, 0) {};
\node[nodeStyle] (v14) at (0.3, 0) {};
\node[nodeStyle] (v15) at (0.6, 0) {};
\node[nodeStyle] (v16) at (0.9, 0) {};
\draw (v0) -- (v11);
\draw (v0) -- (v12);
\draw (v0) -- (v13);
\draw (v0) -- (v14);
\draw (v0) -- (v15);
\draw (v0) -- (v16);
\node[dotStyle] (d0) at (-0.1, 0) {};
\node[dotStyle] (d1) at (0, 0) {};
\node[dotStyle] (d1) at (0.1, 0) {};

\node[nodeStyle] (v17) at (7.7, 0) {};
\node[nodeStyle] (v18) at (7.4, 0) {};
\node[nodeStyle] (v19) at (7.1, 0) {};
\node[nodeStyle] (v20) at (8.3, 0) {};
\node[nodeStyle] (v21) at (8.6, 0) {};
\node[nodeStyle] (v22) at (8.9, 0) {};
\draw (v10) -- (v17);
\draw (v10) -- (v18);
\draw (v10) -- (v19);
\draw (v10) -- (v20);
\draw (v10) -- (v21);
\draw (v10) -- (v22);
\node[dotStyle] (d0) at (7.9, 0) {};
\node[dotStyle] (d1) at (8, 0) {};
\node[dotStyle] (d1) at (8.1, 0) {};

\draw[brc](a0) -- node[below=2cm, left=0.5cm] {} (a1);
\node [label={[label distance=0cm]30:}] at (1.1, -2.3) {$r$ vertices};

\draw[brc](a3) -- node[above=0.4cm] {$r^2$ vertices} (a2);

\node [label={[label distance=0cm]30:}] at (4, -2.5) {$c$};

\node [label={[label distance=0cm]30:}] at (-0.35, -0.82) {$e_1$};
\node [label={[label distance=0cm]30:}] at (8.4, -0.82) {$e_2$};

\end{tikzpicture}
            \caption{Sketch of a graph $G_r$ of the family $(G_r)_{r \in \mathbb{N}}$ used to prove that \greedy\ does not have an approximation guarantee.}
            \label{fig:greedy-counter}
        \end{figure}

    \end{proof}
    The proof is of course independent of the numerator $x$ used to define $c(S) = x / f(S)$.
    Note that while we assumed $k=2$, similar families of graphs can be constructed for any integer constant $k > 2$: Let a connected component of $G_r - c$ be called a \emph{flower}.
    By starting with $c$ and attaching $k$ flowers to $c$, we obtain graphs to which the same argument applies.
    Since $c$ is always chosen as the first vertex, there is at least one flower from which no vertex is chosen, which incurs a cost in $\Theta(r^3)$.
    An optimal strategy chooses the $e_i, i = 1,\dots,k$ vertices and therefore only incurs a cost in $\Theta(r^2)$.

    Gong~et~al.~\cite{DBLP:journals/tcs/GongNSFDS21} analyze the simple greedy algorithm for non-negative monotone set functions with \emph{generic submodularity ratio} $\gamma$, a scalar measuring how close a function is to being submodular with $\gamma \in [0,1]$, where $\gamma = 1$ if and only if the function is submodular.
    While Bergamini~et~al.~\cite{DBLP:journals/corr/abs-1710-01144} already showed that $c(\cdot)$ is not submodular, i.e. $\gamma < 1$, our result actually implies the following stronger result.

    \begin{corollary}
        The group closeness centrality function $c(\cdot)$ can have a submodularity ratio $\gamma$ arbitrarily close to zero.
    \end{corollary}
    \begin{proof}
        Let $g$ be a non-negative monotone set function with generic submodularity ratio $\gamma$.
        Let $S_G$ be the set computed by the simple greedy algorithm with $k \in \mathbb N$.
        Gong~et~al.~\cite{DBLP:journals/tcs/GongNSFDS21} show that the simple greedy algorithm obtains the approximation guarantee $g(S_G) \geq (1-e^{-\gamma}) g(S^*)$, where $S^*$ denotes the optimal solution.
        By Theorem~\ref{greedy_theorem}, in the expression $c(S_G) = \alpha \cdot c(S^*)$, $\alpha \in \mathbb R$ may get arbitrarily close to zero, implying the same for $\gamma$.
    \end{proof}

    \section{Experimental Evaluation}\label{sec:experiments}
    We experimentally evaluate our developed techniques to assess their impact on real-world graphs.
    First, we evaluate each technique on its own (Section~\ref{subsec:abs-vert-and-imp-dv}) before comparing \textsc{Grover} to current state-of-the-art solvers (Section~\ref{subsec:comparison-to-sota}).
    In Section~\ref{subsec:experiments-approximation}, we evaluate the impact of using dominating vertices in \textsc{GS-LS-C}~\cite{angriman2021group} to restrict the search space.

    \textbf{Data and Experimental Setup.}
    We use 48 graphs from Konect~\cite{konect} and the Network Data Repository~\cite{rossi2015network}, covering diverse real-world networks.
    A subset of them was also used in the experimental analysis of \textsc{ILPind}~\cite{staus2023exact}.
    They range from small graphs with about 50 vertices to graphs with more than $50\,000$ vertices.
    Table~\ref{tab:graph-stats-compact} in the Appendix holds a complete list with detailed statistics.
    The experiments are run on an Intel Xeon Silver 4216 (2.10GHz, 92 GiB RAM, Ubuntu 24.04.3).
    We test $k \in \{2,\ldots,20\}$ with a 10-minute limit per instance, totaling 912 instances.
    For $k=1$, \greedy\ already produces the exact solution. Therefore, unlike \textsc{ILPind}, \textsc{Grover} computes the solution using \greedy\ for $k=1$. Hence, for a fair comparison, we exclude $k=1$ and report runtimes (excluding I/O) averaged over five runs.

    \textbf{Solvers.}
    \textsc{ILPind}~\cite{staus2023exact} is the state-of-the-art algorithm for \gccm.
    We modified it to use Gurobi 12.0.1~\cite{gurobi}.
    \textsc{DVind}~\cite{staus2023exact} is the fastest explicit branch-and-bound algorithm to solve \gccm\ and it also deploys dominating vertices.
    Both \textsc{ILPind} and \textsc{DVind} are implemented in Kotlin and run using Java openjdk 21.0.8.
    \textsc{SubModST}~\cite{Woydt24} is a branch-and-bound algorithm for general submodular function maximization.
    It employs more pruning techniques and optimizations, but does not utilize dominating vertices.
    It is implemented in C++17.
    Our algorithm \textsc{Grover} also utilizes Gurobi 12.0.1 and is implemented in C++20.
    In the approximate setting we compare against \textsc{GS-LS-C}~\cite{angriman2023algorithms}, which is implemented in the C++ library NetworKit~\cite{angriman2023algorithms}.
    All C++ solvers are compiled using gcc 15.1.0 and \hbox{\textsc{-O3}} optimization.
    Note that \textsc{PRESTO}~\cite{rajbhandari2022presto} is not publicly available and their experiments show that it does not scale for $k > 7$.

    \subsection{Absorbed Vertices and Improved $d(v)$}\label{subsec:abs-vert-and-imp-dv}
    In this section, we evaluate the effect of absorbing vertices and initially approximating the $d(v)$-values.
    Recall that absorbing vertices is a data reduction technique that allows us to remove variables and corresponding constraints from the ILPs.
    Approximating the $d(v)$-values should overall result in fewer iterations of the algorithm and therefore improve the running time.
    Additionally, we compare the runtime of \textsc{Grover} when it is initialized by \textsc{GS-LS-C}~\cite{angriman2023algorithms} and when it is initialized with an optimal solution.
    This provides insight into how much more speedup one can expect by utilizing better initial solutions.

    \begin{figure}
        \centering
        \begin{minipage}[t]{0.48\textwidth}
            \centering
            \includegraphics[width=\linewidth]{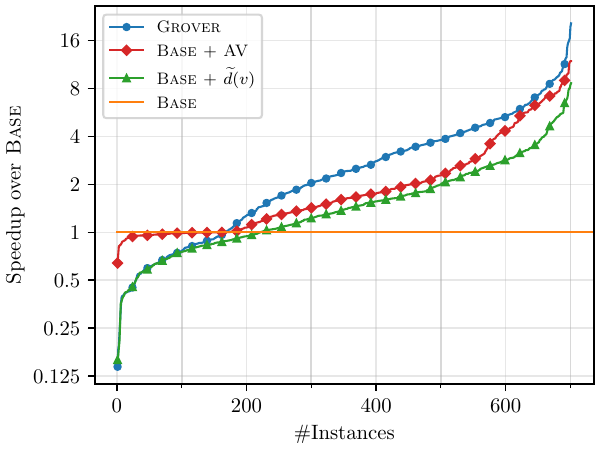}

            \small (a)
        \end{minipage}\hfill
        \begin{minipage}[t]{0.48\textwidth}
            \centering
            \includegraphics[width=\linewidth]{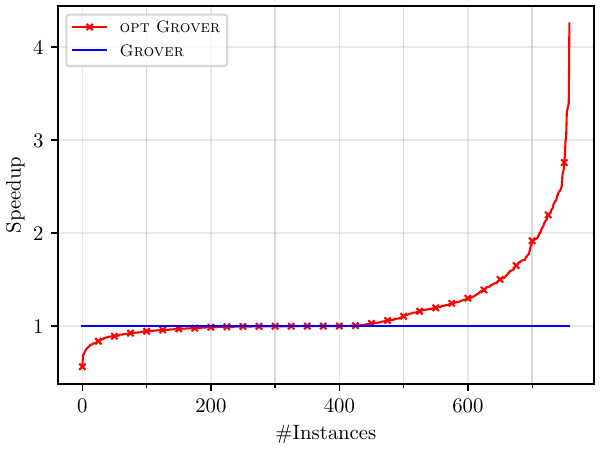}

            \small (b)
        \end{minipage}
        \caption{
            (a) Speedup of our different optimizations over \textsc{Base}. \textsc{Base} uses dominating vertices and initially sets $d(v)=2$. Absorbing vertices \textsc{AV} and $\widetilde{d}(v)$ are evaluated separately. \textsc{Grover} uses both.
            (b) Speedup of \textsc{Grover} initialized with an optimal solution compared to the default initialization with \textsc{GS-LS-C}~\cite{angriman2023algorithms}. Speedups are determined per instance and then sorted.
        }
        \label{fig:figure2}
    \end{figure}

    First, we compare the effect of absorbing vertices (\textsc{AV}) and approximating the $d(v)$-values.
    The \textsc{Base} algorithm only uses dominating vertices and the default $d(v) = 2$ initialization.
    We compare to \textsc{Base + AV} which additionally enables absorbing vertices, but keeps $d(v) = 2$.
    The algorithm \textsc{Base + $\widetilde{d}(v)$} approximates the $d(v)$-values, but does not enable absorbing vertices.
    Finally, \textsc{Grover} enables absorbing vertices and approximates the $d(v)$-values.

    Figure~\ref{fig:figure2}a shows the speedup of each configuration over \textsc{Base}.
    For each instance the speedup is determined and these are sorted increasingly.
    Note that these speedups are determined on the 701 instances all configurations are able to solve.
    \hbox{\textsc{Base + AV}}, \textsc{Base + $\widetilde{d}(v)$} and \textsc{Grover} solve 713, 750 and 759 instances respectively.
    Table~\ref{tab:config-comparison} in the Appendix shows more detailed results for each graph.
    \textsc{Base} is overall the slowest configuration as the other configurations have a mean speedup of $1.18\times$, $1.53\times$ and $2.19\times$ respectively.
    Note that these mean speedups are however skewed by many easy-to-solve instances, for which the overhead of computing \textsc{GS-LS-C} has a detrimental effect on the total runtime.
    This can be seen for \textsc{Base + $\widetilde{d}(v)$} and \textsc{Grover} in Figure~\ref{fig:figure2}a.
    Therefore, we also compare the speedup measured as total sum of time needed to solve all 701 instances, and simply refer to this measure as \emph{speedup}, while we refer to the geometric mean speedup as \emph{mean speedup}.
    Then, relative to \textsc{Base}, \hbox{\textsc{Base + AV}} has a speedup of $2.07\times$, while \hbox{\textsc{Base + $\widetilde{d}(v)$}} has a speedup of $1.41\times$.
    \textsc{Grover} has a speedup of $3.20\times$.
    For harder-to-solve instances that require more time, our improvements have a greater effect in reducing the runtime.
    While \textsc{Grover} has the largest speedup, it is not the fastest configuration for each instance.
    In total, 760 instances were solved by at least one configuration.
    Out of these, \textsc{Grover} is the fastest on 411, while \textsc{Base + AV} is the fastest on 185 instances.
    \textsc{Base} is the fastest configuration on 109, while \hbox{\textsc{Base + $\widetilde d(v)$}} is the fastest on 55 instances.
    For larger instances, \textsc{Grover} is the fastest algorithm, however for smaller instances, there is no clear winner.

    Secondly, we compare \textsc{Grover} when initialized with \textsc{GS-LS-C} versus an optimal solution, to assess how much room for improvement remains in the initialization.
    A large difference between the running times indicates that better initial solutions will improve the running time of \textsc{Grover}.
    However, small differences show that the running time cannot be further improved by other initializations.
    Note that initializing \textsc{Grover} with an optimal solution means that the algorithm essentially only verifies that the solution is indeed optimal.
    However, it does not imply that only one iteration of the ILP is necessary.
    While the ILP is sufficient for the optimal input solution, it may not be sufficient for another optimal solution.
    To verify that the other solution is indeed not strictly better the algorithm may need to enlarge the ILP, which could result in larger running times.
    For this experiment, we removed the time it costs \textsc{Grover} to run \textsc{GS-LS-C}, since \textsc{opt Grover} has no additional cost to obtain the optimal solution.
    Therefore, only the time it costs to solve the ILPs is measured.
    Figure~\ref{fig:figure2}b shows the corresponding speedups, while Table~\ref{tab:config-comparison} in the Appendix shows a more detailed view.
    The speedups are computed per instance and sorted afterwards.
    Out of the 759 instances, 401 are solved faster when initialized with an optimal solution while 358 are solved slower.
    The smallest speedup achieved is $0.56\times$ for \textsc{inf-euroroad} and $k=2$, however the mean speedup across all slower instances is $0.95\times$, the loss in time therefore is negligible.
    Overall, the speedup is $1.14\times$ while the mean speedup is $1.13\times$.
    For 46 instances a speedup of greater than 2 and for 7 instances a speedup of greater than 3 is achieved.
    Initialization with an optimal solution does not greatly improve the overall running time of the ILP, indicating that the approximate initialization with \textsc{GS-LS-C}~\cite{angriman2023algorithms} is already of high quality.

    \subsection{Comparison to SOTA}\label{subsec:comparison-to-sota}
    We compare our algorithm \textsc{Grover} against \textsc{ILPind}~\cite{staus2023exact}, \textsc{DVind}~\cite{staus2023exact} and \textsc{SubModST}~\cite{Woydt24}.
    Figure~\ref{fig:figure}a shows the number of instances each algorithm is able to solve.
    Table~\ref{tab:sota-comparison} in the Appendix holds more detailed running times.
    While \textsc{SubModST} performs very fast for the first few instances, it is overall the slowest algorithm and only solves 329 instances.
    \textsc{DVind} performs better as it solves 546 instances, but still falls behind both of the iterative ILP algorithms.
    \textsc{ILPind} solves 692 instances while \textsc{Grover} solves 759 out of the 912 possible instances.
    \textsc{Grover} has a mean speedup of $4.59\times$ compared to \textsc{ILPind} and a maximum speedup of $34.16\times$.
    Only on 22 out of the 692 instances solved by both solvers is \textsc{Grover} slower.
    In this case, \textsc{Grover}'s worst speedup is $0.40\times$ while the mean speedup on those 22 instances is $0.74\times$.
    However, each instance solved by \textsc{ILPind} is also solved by \textsc{Grover}.
    Overall, \textsc{Grover} spends more than $95\%$ of its runtime in solving the ILPs.

    \begin{figure}[t!]
        \centering
        \begin{minipage}[t]{0.48\textwidth}
            \centering
            \includegraphics[width=\linewidth]{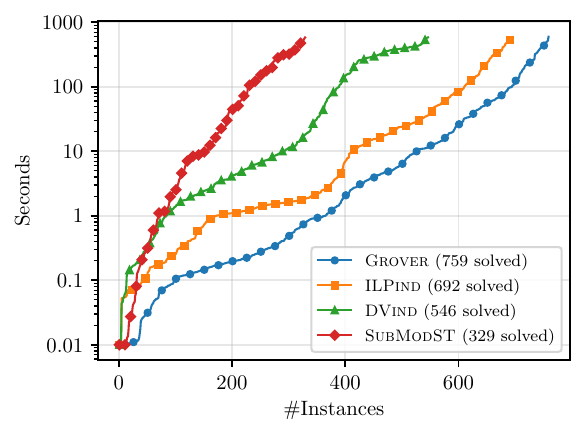}

            \small (a)
        \end{minipage}\hfill
        \begin{minipage}[t]{0.48\textwidth}
            \centering
            \includegraphics[width=\linewidth]{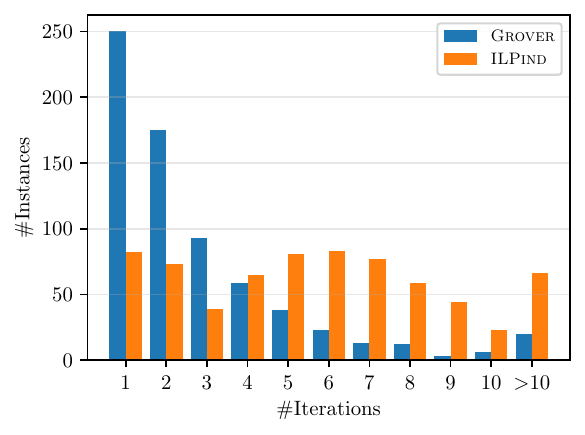}

            \small (b)
        \end{minipage}
        \caption{
            (a) Number of solved instances for \textsc{Grover}, \textsc{ILPind}~\cite{staus2023exact}, \textsc{DVind}~\cite{staus2023exact}, and \textsc{SubModST}~\cite{Woydt24}.
            (b) Distribution of the number of iterations that \textsc{Grover} (ours) and \textsc{ILPind}~\cite{staus2023exact} need to solve all instances.
        }
        \label{fig:figure}
    \end{figure}

    Figure~\ref{fig:figure}b shows how many instances \textsc{Grover} and \textsc{ILPind} solve and how many ILP iterations are necessary.
    This plot only includes the 692 instances that both \textsc{Grover} and \textsc{ILPind} are able to solve.
    \textsc{Grover} is able to solve significantly more instances in fewer iterations than \textsc{ILPind}.
    About $36\%$ of instances are solved within one iteration, while \textsc{ILPind} only solves about $11\%$ within one iteration.
    For more than $50\%$ of instances, \textsc{ILPind} needs 6 or more iterations while \textsc{Grover} needs 6 or more iterations for only about $11\%$ of all instances.
    Note that there is one instance, \textsc{econ-mahindas} and $k=20$, where \textsc{ILPind} only requires three iterations, while \textsc{Grover} needs four.
    In total, \textsc{Grover} greatly reduces the number of needed iterations, which was the goal of approximating the $d(v)$-values.

    Figure~\ref{fig:approx_figure}a shows the geometric mean, minimum, and maximum speedup of \textsc{Grover} over \textsc{ILPind} across $k$.
    For each $k$ all instances solved by both algorithms are included in this comparison.
    The mean speedup peaks at $7.91\times$ for $k=2$ and decreases to $3.80\times$ for $k=20$, with a similar trend for min and max values.
    The maximum speedup reaches $34.16\times$ at $k=3$ and drops to $9.25\times$ at $k=20$.
    For $k \leq 11$ (except $k=5$), the minimum speedup exceeds 1, meaning \textsc{Grover} is never slower, while for $k \geq 12$ it can be slower, with a minimum of $0.40$ at $k=19$.
    Note however that \textsc{Grover} is only slower on 22 out of 692 instances.

    \textbf{Addressing Implementation Differences.}
    Since the implementation details of \textsc{Grover} and \textsc{ILPind} differ, for example with respect to the programming language used, a proper comparison has to take these differences into account.
    The algorithm \textsc{Base} is in essence a re-implementation of \textsc{ILPind}.
    \textsc{Grover} is built on top of \textsc{Base} by extending it with our improvements.
    When comparing \textsc{Base} to \textsc{ILPind} using the sum of runtimes over all instances (solved by both algorithms), \textsc{Base} achieves a modest $1.15 \times$ speedup.
    We use this measure rather than the geometric mean of speedups to minimize the skew of JVM warmup (among other effect due to implementation differences) on instances solved quickly.
    This comparison ensures that the observed speedups can be strictly attributed to our algorithmic improvements, rather than implementation details.

    \subsection{Dominating Vertices in \textsc{GS-LS-C}}\label{subsec:experiments-approximation}
    Since using dominating vertices provably restricts the search space in the exact case, we also want to empirically explore the search space restriction in an approximate setting, which as discussed in Section \ref{sec:approximation-results} does not compromise its approximation guarantee.
    We therefore compare the running time and solution quality of \textsc{GS-LS-C}~\cite{angriman2021group} with and without the restriction.
    We test the values of $k \in \{10, 50, 100\}$ on each graph with more than 100 vertices.
    These group sizes were also used in the experimental evaluation of \textsc{GS-LS-C}.
    Each graph-$k$ instance is run 5 times and the required time and achieved group farness are averaged.
    Figure~\ref{fig:approx_figure}b shows the speedups.
    For $k=10$, the mean speedup is $1.54\times$ (min $0.69\times$, max $5.83\times$), with slowdowns on small instances due to the overhead of identifying dominated vertices.
    This overhead diminishes for larger instances: for $k=50$ and $k=100$, mean speedups increase to $1.59\times$ and $1.71\times$, with minima $0.85\times$, $0.90\times$ and maxima $7.51\times$, $9.64\times$.
    The solution quality largely remains unaffected.
    In fact, the average increase in group farness is smaller than $0.01\%$ for all values of $k$.
    Overall, restricting the search space in \textsc{GS-LS-C}~\cite{angriman2021group} improves the running time without altering solution quality.
    Including the restriction when the approximation algorithm is initialized with different initial solutions will likely lead to the same results.

    \begin{figure}
        \centering
        \begin{minipage}{0.48\textwidth}
            \centering
            \includegraphics[width=\linewidth]{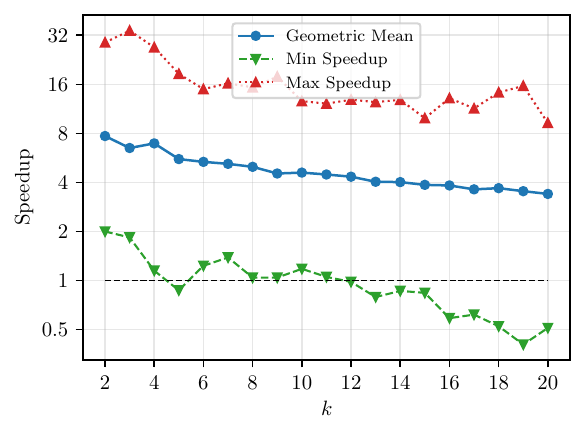}

            \small (a)
        \end{minipage}\hfill
        \begin{minipage}{0.48\textwidth}
            \centering
            \includegraphics[width=\linewidth]{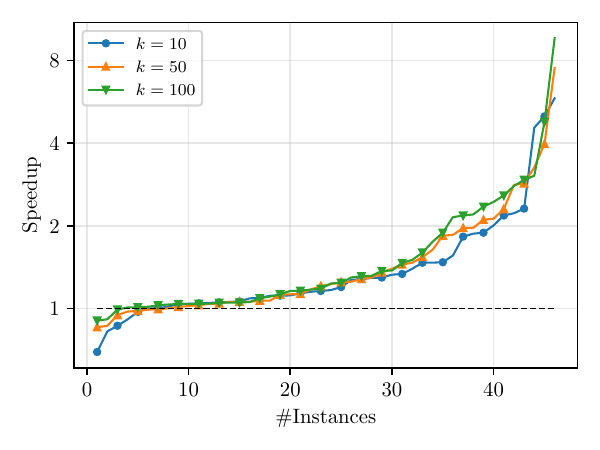}

            \small (b)
        \end{minipage}

        \caption{
            (a) Speedup of \textsc{Grover} over \textsc{ILPind}~\cite{staus2023exact} for different values of $k$.
            (b) Runtime ratio of \textsc{GS-LS-C}~\cite{angriman2021group} without dominating vertices and with dominating vertices.}
        \label{fig:approx_figure}
    \end{figure}

    \section{Conclusion and Future Work}\label{sec:conclusion}
    Group closeness centrality maximization asks for a set of $k$ vertices, such that the distance to all other vertices in the graph is minimized.
    In this work, we introduced two new techniques for the exact case.
    The first one, absorbed vertices, proves that certain vertices have their distance to any optimal solution determined by a neighbor.
    This leads to a smaller ILP as the variables associated with the absorbed vertices can be removed.
    The second technique bootstraps the ILP using an approximate solution to produce near-sufficient ILPs.
    This leads the algorithm to solve many instances in only one iteration and overall greatly reduces the number of required iterations.
    Our algorithm \textsc{Grover} has a mean speedup of $4.59\times$ and can achieve speedups of up to $34.16\times$ over the state-of-the-art algorithm \textsc{ILPind}~\cite{staus2023exact}.
    We proved that restricting the search space of \textsc{GS-LS-C}~\cite{angriman2023algorithms} to non-dominated vertices preserves its $1/5$-approximation guarantee, achieving speedups of up to $9.64\times$.
    Additionally, we settled an open question by proving that a widely-used greedy algorithm~\cite{chen2016efficient} does not give a constant factor approximation.
    Future work could explore if the techniques proposed here can be applied to more general problems, such as the weighted version of the problem studied here.

    \bibliography{bibliography}

@inproceedings{staus2023exact,
  author = {Luca Pascal Staus and Christian Komusiewicz and Nils Morawietz and Frank Sommer},
  editor = {Jonathan W. Berry and David B. Shmoys and Lenore Cowen and Uwe Naumann},
  title = {{E}xact {A}lgorithms for {G}roup {C}loseness {C}entrality},
  booktitle = {{SIAM} Conference on Applied and Computational Discrete Algorithms, {ACDA} 2023, Seattle, WA, USA, May 31 - June 2, 2023},
  pages = {1--12},
  publisher = {{SIAM}},
  year = {2023},
  doi = {10.1137/1.9781611977714.1}
}

@article{nemhauser1978analysis,
  author = {George L. Nemhauser and Laurence A. Wolsey and Marshall L. Fisher},
  title = {An analysis of approximations for maximizing submodular set functions - {I}},
  journal = {Math. Program.},
  volume = {14},
  number = {1},
  pages = {265--294},
  year = {1978},
  doi = {10.1007/BF01588971},
}

@InProceedings{chen2016efficient,
  author    = {Chen Chen and Wei Wang and Xiaoyang Wang},
  booktitle = {Databases Theory and Applications - 27th Australasian Database Conference, {ADC} 2016, Sydney, NSW, Australia, September 28-29, 2016, Proceedings},
  title     = {{E}fficient {M}aximum {C}loseness {C}entrality {G}roup {I}dentification},
  doi       = {10.1007/978-3-319-46922-5_4},
  editor    = {Muhammad Aamir Cheema and Wenjie Zhang and Lijun Chang},
  pages     = {43--55},
  publisher = {Springer},
  series    = {Lecture Notes in Computer Science},
  volume    = {9877},
  year      = {2016},
}

@article{DBLP:journals/corr/abs-1710-01144,
  author = {Elisabetta Bergamini and Tanya Gonser and Henning Meyerhenke},
  title = {{S}caling up {G}roup {C}loseness {M}aximization},
  journal = {CoRR},
  volume = {abs/1710.01144},
  year = {2017},
  doi = {10.48550/arXiv.1710.01144}
}

@InCollection{angriman2023algorithms,
  author    = {Eugenio Angriman and Alexander van der Grinten and Michael Hamann and Henning Meyerhenke and Manuel Penschuck},
  booktitle = {Algorithms for Big Data - {DFG} Priority Program 1736},
  title     = {{A}lgorithms for {L}arge-scale {N}etwork {A}nalysis and the {N}etwor{K}it {T}oolkit},
  doi       = {10.1007/978-3-031-21534-6_1},
  editor    = {Hannah Bast and Claudius Korzen and Ulrich Meyer and Manuel Penschuck},
  pages     = {3--20},
  publisher = {Springer},
  series    = {Lecture Notes in Computer Science},
  volume    = {13201},
  year      = {2022},
}

@article{everett1999centrality,
  author = {M. G. Everett and S. P. Borgatti},
  title = {The centrality of groups and classes},
  journal = {The Journal of Mathematical Sociology},
  volume = {23},
  number = {3},
  pages = {181--201},
  year = {1999},
  publisher = {Routledge},
  doi = {10.1080/0022250X.1999.9990219}
}

@InProceedings{rossi2015network,
  author    = {Rossi, Ryan A. and Ahmed, Nesreen K.},
  booktitle = {Proceedings of the Twenty-Ninth {AAAI} Conference on Artificial Intelligence, January 25-30, 2015, Austin, Texas, {USA}},
  title     = {{T}he {N}etwork {D}ata {R}epository with {I}nteractive {G}raph {A}nalytics and {V}isualization},
  doi       = {10.1609/AAAI.V29I1.9277},
  editor    = {Blai Bonet and Sven Koenig},
  pages     = {4292--4293},
  publisher = {{AAAI} Press},
  year      = {2015},
}

@InProceedings{DBLP:journals/corr/abs-1911-03360,
  author    = {Eugenio Angriman and Alexander van der Grinten and Henning Meyerhenke},
  booktitle = {2019 {IEEE} International Conference on Big Data {(IEEE} BigData), Los Angeles, CA, USA, December 9-12, 2019},
  title     = {{L}ocal {S}earch for {G}roup {C}loseness {M}aximization on {B}ig {G}raphs},
  doi       = {10.1109/BIGDATA47090.2019.9006206},
  editor    = {Chaitanya K. Baru and Jun Huan and Latifur Khan and Xiaohua Hu and Ronay Ak and Yuanyuan Tian and Roger S. Barga and Carlo Zaniolo and Kisung Lee and Yanfang (Fanny) Ye},
  pages     = {711--720},
  publisher = {{IEEE}},
  bibsource = {dblp computer science bibliography, https://dblp.org},
  year      = {2019},
}

@InProceedings{angriman2021group,
  author    = {Eugenio Angriman and Ruben Becker and Gianlorenzo D'Angelo and Hugo Gilbert and Alexander van der Grinten and Henning Meyerhenke},
  booktitle = {Proceedings of the 23rd Symposium on Algorithm Engineering and Experiments, {ALENEX} 2021, Virtual Conference, January 10-11, 2021},
  title     = {{G}roup-{H}armonic and {G}roup-{C}loseness {M}aximization - {A}pproximation and {E}ngineering},
  doi       = {10.1137/1.9781611976472.12},
  editor    = {Martin Farach{-}Colton and Sabine Storandt},
  pages     = {154--168},
  publisher = {{SIAM}},
  bibsource = {dblp computer science bibliography, https://dblp.org},
  year      = {2021},
}

@Article{10f5227e-a114-31a3-b1b2-cc4f421e84a5,
  author    = {Alex Bavelas},
  title     = {A MATHEMATICAL MODEL FOR GROUP STRUCTURES},
  issn      = {0093-2914},
  number    = {3},
  pages     = {16--30},
  url       = {http://www.jstor.org/stable/44135428},
  urldate   = {2025-08-23},
  volume    = {7},
  journal   = {Applied Anthropology},
  publisher = {Society for Applied Anthropology},
  year      = {1948},
}

@Article{https://doi.org/10.1002/bs.3830100205,
  author  = {Beauchamp, Murray A.},
  title   = {An improved index of centrality},
  doi     = {10.1002/bs.3830100205},
  number  = {2},
  pages   = {161--163},
  volume  = {10},
  journal = {Behavioral Science},
  year    = {1965},
}

@article{sabidussi1966centrality,
  author  = {Gert Sabidussi},
  title   = {The Centrality Index of a Graph},
  journal = {Psychometrika},
  volume  = {31},
  number  = {4},
  pages   = {581--603},
  year    = {1966},
  doi     = {10.1007/BF02289527}
}

@Book{social2013,
  author    = {Borgatti, Stephen P. and Everett, Martin G. and Johnson, Jeffrey C.},
  title     = {Analyzing social networks},
  isbn      = {9781446247419},
  publisher = {{SAGE}},
  year      = {2013},
}

@Article{van2013network,
  author  = {Martijn P. van den Heuvel and Olaf Sporns},
  title   = {Network Hubs in the Human Brain},
  doi     = {10.1016/j.tics.2013.09.012},
  number  = {12},
  pages   = {683--696},
  volume  = {17},
  journal = {Trends in Cognitive Sciences},
  year    = {2013},
}

@Article{chea2007accurate,
  author  = {Chea, Eric and Livesay, Dennis R.},
  title   = {How accurate and statistically robust are catalytic site predictions based on closeness centrality?},
  doi     = {10.1186/1471-2105-8-153},
  volume  = {8},
  journal = {{BMC} Bioinform.},
  year    = {2007},
}

@Article{AdebayoSun+2020,
  author  = {Isaiah G. Adebayo and Yanxia Sun},
  title   = {A novel approach of closeness centrality measure for voltage stability analysis in an electric power grid},
  doi     = {10.1515/ijeeps-2020-0013},
  number  = {3},
  pages   = {20200013},
  volume  = {21},
  journal = {International Journal of Emerging Electric Power Systems},
  year    = {2020},
}

@inbook{Koschuetzki2005,
  author = {Kosch{\"u}tzki, Dirk and Lehmann, Katharina Anna and Peeters, Leon and Richter, Stefan and Tenfelde-Podehl, Dagmar and Zlotowski, Oliver},
  editor = {Brandes, Ulrik and Erlebach, Thomas},
  title = {{C}entrality {I}ndices},
  booktitle = {{N}etwork {A}nalysis: {M}ethodological {F}oundations},
  year = {2005},
  publisher = {Springer Berlin Heidelberg},
  address = {Berlin, Heidelberg},
  pages = {16--61},
  isbn = {978-3-540-31955-9},
  doi = {10.1007/978-3-540-31955-9_3}
}

@Article{arya2001local,
  author  = {Arya, Vijay and Garg, Naveen and Khandekar, Rohit and Meyerson, Adam and Munagala, Kamesh and Pandit, Vinayaka},
  title   = {{L}ocal {S}earch {H}euristics for k-{M}edian and {F}acility {L}ocation {P}roblems},
  doi     = {10.1137/S0097539702416402},
  number  = {3},
  pages   = {544--562},
  volume  = {33},
  journal = {{SIAM} J. Comput.},
  year    = {2004},
}

@article{rajbhandari2022presto,
  author = {Baibhav Rajbhandari and Paul Olsen Jr. and Jeremy Birnbaum and Jeong{-}Hyon Hwang},
  title = {{PRESTO:} {F}ast and {E}ffective {G}roup {C}loseness {M}aximization},
  journal = {{IEEE} Trans. Knowl. Data Eng.},
  volume = {35},
  number = {6},
  pages = {6209--6223},
  year = {2023},
  doi = {10.1109/TKDE.2022.3178925}
}

@InProceedings{Woydt24,
  author             = {Henning Martin Woydt and Christian Komusiewicz and Frank Sommer},
  booktitle          = {32nd Annual European Symposium on Algorithms, {ESA} 2024, Royal Holloway, London, United Kingdom, September 2-4, 2024},
  title              = {{S}ub{M}od{ST}: {A} {F}ast {G}eneric {S}olver for {S}ubmodular {M}aximization with {S}ize {C}onstraints},
  doi                = {10.4230/LIPIcs.ESA.2024.102},
  editor             = {Timothy M. Chan and Johannes Fischer and John Iacono and Grzegorz Herman},
  pages              = {102:1--102:18},
  publisher          = {Schloss Dagstuhl - Leibniz-Zentrum f{\"{u}}r Informatik},
  series             = {LIPIcs},
  volume             = {308},
  _bib2doi_confirmed = {true},
  _bib2doi_selected  = {dblp:/rec/conf/esa/WoydtKS24.bib},
  bibsource          = {dblp computer science bibliography, https://dblp.org},
  biburl             = {https://dblp.org/rec/conf/esa/WoydtKS24.bib},
  timestamp          = {Tue, 24 Sep 2024 01:00:00 +0200},
  year               = {2024},
}

@inproceedings{Rymon92,
  author = {Ron Rymon},
  editor = {Bernhard Nebel and Charles Rich and William R. Swartout},
  title = {{S}earch through {S}ystematic {S}et {E}numeration},
  booktitle = {Proceedings of the 3rd International Conference on Principles of Knowledge Representation and Reasoning (KR'92). Cambridge, MA, USA, October 25-29, 1992},
  pages = {539--550},
  publisher = {Morgan Kaufmann},
  year = {1992}
}

@Manual{gurobi,
  author = {{Gurobi Optimization, LLC}},
  title  = {{G}urobi {O}ptimizer {R}eference {M}anual},
  note   = {Available at},
  url    = {https://www.gurobi.com},
  year   = {2025},
}

@inproceedings{konect,
  title = {{KONECT} -- {The} {Koblenz} {Network} {Collection}},
  author = {J\'{e}r\^{o}me Kunegis},
  year = {2013},
  booktitle = {Proc. Int. Conf. on World Wide Web Companion},
  pages = {1343--1350},
}

@Thesis{ternes_2025,
  author = {Ternes, Jakob},
  title  = {{A}lgorithms for {G}roup {C}loseness {C}entrality {M}aximization},
  doi    = {10.5281/zenodo.18293375},
  month  = dec,
  year   = {2025},
}

@Article{DBLP:journals/tcs/GongNSFDS21,
  author    = {Suning Gong and Qingqin Nong and Tao Sun and Qizhi Fang and Ding{-}Zhu Du and Xiaoyu Shao},
  title     = {Maximize a monotone function with a generic submodularity ratio},
  doi       = {10.1016/J.TCS.2020.05.018},
  pages     = {16--24},
  volume    = {853},
  bibsource = {dblp computer science bibliography, https://dblp.org},
  journal   = {Theor. Comput. Sci.},
  year      = {2021},
}

@Article{Tarjan72,
  author    = {Robert Endre Tarjan},
  title     = {Depth-First Search and Linear Graph Algorithms},
  doi       = {10.1137/0201010},
  number    = {2},
  pages     = {146--160},
  volume    = {1},
  bibsource = {dblp computer science bibliography, https://dblp.org},
  journal   = {{SIAM} J. Comput.},
  year      = {1972},
}

    \newpage

    \clearpage

    \section*{Appendix}

    \section{Graphs and Results}\label{app:tables}
        {\centering
    \captionof{table}{Graphs used in the experimental evaluation.
    \#dom and \#abs denote the numbers of dominated and absorbed vertices, respectively;
    dens is the graph density and diam its diameter.
    }
    \resizebox{\linewidth}{!}{%
        \begin{tabular}{lrrrrrr|lrrrrrr}
            \hline
            Graph                          & $n$    & $m$      & dens   & diam & \#dom  & \#abs  & Graph               & $n$     & $m$         & dens   & diam & \#dom   & \#abs   \\
            \hline
            contiguous-usa                 & 49     & 107      & 0.091  & 11   & 13     & 1      & fb-pages-tvshow     & 3\,892  & 17\,239     & 0.002  & 20   & 1\,690 & 688 \\
            brain\_1 & 65     & 730      & 0.351  & 3    & 17     & 0      & ca-GrQc             & 4\,158  & 13\,422     & 0.002  & 17 & 2\,712 & 1\,171 \\
            arenas-jazz                    & 198    & 2\,742   & 0.141  & 6    & 93     & 5      & inf-power           & 4\,941  & 6\,594      & <0.001 & 46   & 1\,487  & 1\,278  \\
            ca-netscience                  & 379    & 914      & 0.013  & 17   & 306    & 128    & ca-Erdos992         & 4\,991  & 7\,428      & <0.001 & 14   & 3\,861 & 3\,516 \\
            robot24c1\_mat5                & 404    & 14\,261  & 0.175  & 3    & 10     & 0      & soc-advogato        & 5\,054  & 39\,374     & 0.003  & 9    & 1\,721 & 1\,080 \\
            tortoise   & 496    & 984      & 0.008  & 21   & 232    & 140    & fb-pages-politician & 5\,908 & 41\,706 & 0.002 & 14 & 1\,922 & 644 \\
            econ-beause                    & 507    & 39\,427  & 0.307  & 3    & 490    & 1      & ia-reality          & 6\,809  & 7\,680      & <0.001 & 8    & 6\,391  & 6\,284  \\
            bio-diseasome                  & 516    & 1\,188   & 0.009  & 15   & 379    & 204    & bio-dmela           & 7\,393  & 25\,569     & <0.001 & 11   & 2\,042 & 2\,005 \\
            soc-wiki-Vote                  & 889    & 2\,914   & 0.007  & 13   & 239    & 193    & rt\_onedirection    & 7\,987  & 8\,100      & <0.001 & 11 & 7\,725 & 7\,717 \\
            ca-CSphd                       & 1\,025 & 1\,043   & 0.002  & 28   & 697    & 693    & bio-grid-human      & 9\,186  & 31\,038     & <0.001 & 14   & 3\,425 & 2\,843 \\
            inf-euroroad                   & 1\,039 & 1\,305   & 0.002  & 62   & 132    & 127    & rt\_lolgop          & 9\,765  & 10\,075     & <0.001 & 10   & 9\,282 & 9\,142 \\
            email-univ                     & 1\,133 & 5\,451   & 0.009  & 8    & 232    & 153    & ia-escorts-dynamic  & 10\,106 & 39\,016     & <0.001 & 10 & 1\,857 & 1\,838 \\
            econ-mahindas                  & 1\,258 & 7\,513   & 0.010  & 8    & 4      & 0      & ca-HepPh            & 11\,204 & 117\,619    & 0.002  & 13   & 7\,014 & 1\,904 \\
            bio-yeast                      & 1\,458 & 1\,948   & 0.002  & 19   & 800    & 735    & web-indochina-2004  & 11\,358 & 47\,606     & <0.001 & 27 & 9\,794 & 7\,234 \\
            comsol                         & 1\,500 & 48\,119  & 0.043  & 12   & 1\,485 & 0      & soc-anybeat         & 12\,645 & 49\,132     & <0.001 & 10   & 8\,430 & 6\,320 \\
            medulla\_1  & 1\,770 & 8\,905   & 0.006  & 6    & 836    & 392    & fb-pages-company    & 14\,113 & 52\,126 & <0.001 & 15 & 4\,568 & 2\,511 \\
            rt\_assad                      & 2\,139 & 2\,786   & 0.001  & 14   & 1\,577 & 1\,562 & econ-poli-large     & 15\,575 & 17\,468     & <0.001 & 15 & 12\,970 & 12\,515 \\
            bio-WormNet-v3       & 2\,274 & 78\,328  & 0.030  & 11   & 2\,080 & 66     & ca-AstroPh          & 17\,903 & 196\,972 & 0.001 & 14 & 10\,462 & 1\,701 \\
            heart2                         & 2\,339 & 340\,229 & 0.124  & 6    & 2\,327 & 0      & ca-CondMat          & 21\,363 & 91\,286     & <0.001 & 15   & 14\,282 & 3\,386 \\
            econ-orani678                  & 2\,529 & 86\,768  & 0.027  & 5    & 42     & 0      & ca-cit-HepTh        & 22\,721 & 2\,444\,642 & 0.009 & 9 & 10\,536 & 317 \\
            road-minnesota                 & 2\,640 & 3\,302   & <0.001 & 99   & 97     & 95     & fb-pages-media      & 27\,917 & 205\,964    & <0.001 & 15 & 6\,330 & 2\,364 \\
            inf-openflights                & 2\,905 & 15\,645  & 0.004  & 14   & 1\,783 & 806    & soc-gemsec-RO       & 41\,773 & 125\,826 & <0.001 & 19 & 6\,626 & 5\,490 \\
            web-edu                        & 3\,031 & 6\,474   & 0.001  & 11   & 2\,768 & 2\,715 & soc-gemsec-HU       & 47\,538 & 222\,887    & <0.001 & 14 & 3\,768 & 2\,725 \\
            econ-psmigr3                   & 3\,140 & 410\,781 & 0.083  & 3    & 227    & 0      & fb-pages-artist     & 50\,515 & 819\,090    & <0.001 & 11 & 4\,967 & 3\,168 \\
            \hline
        \end{tabular}
    }
    \label{tab:graph-stats-compact}}

    \clearpage
    \begin{table}[ht]
    \centering
    \caption{Per-graph comparison of \textsc{SubModST}~\cite{Woydt24}, \textsc{DVind}~\cite{staus2023exact}, \textsc{ILPind}~\cite{staus2023exact} and \textsc{Grover} (ours).
    For each graph, \#Solved reports the number of $k$ values for which an algorithm solved the instance, and $\sum$Seconds is the cumulative runtime over all successfully solved instances.
    Note that different algorithms may solve different subsets of $k$ values.
    Bold entries indicate the algorithm that solved the most instances and, in case of a tie, achieved the lowest total runtime.
    }
    \resizebox{\linewidth}{!}{%
        \begin{tabular}{lrrrrrrrr}
            \hline
            \multirow{2}{*}{Graph} & \multicolumn{2}{c}{\textsc{SubModST}} & \multicolumn{2}{c}{\textsc{DVind}} & \multicolumn{2}{c}{\textsc{ILPind}} & \multicolumn{2}{c}{\textsc{Grover}} \\
            \cline{2-9}
            & \#Solved & $\sum$Seconds & \#Solved   & $\sum$Seconds    & \#Solved & $\sum$Seconds & \#Solved    & $\sum$Seconds    \\
            \hline
            contiguous-usa                 & 19       & 737.35        & 19         & 19.09            & 19       & 1.45          & \textbf{19} & \textbf{0.18}    \\
            brain\_1 & 19       & 110.89        & 19         & 23.51            & 19       & 1.93          & \textbf{19} & \textbf{0.21}    \\
            arenas-jazz                    & 12       & 261.39        & 14         & 678.52           & 19       & 3.82          & \textbf{19} & \textbf{0.63}    \\
            ca-netscience                  & 12       & 899.63        & 19         & 8.39             & 19       & 4.67          & \textbf{19} & \textbf{1.03}    \\
            robot24c1\_mat5                & 8        & 475.52        & 6          & 163.96           & 19       & 15.60         & \textbf{19} & \textbf{7.20}    \\
            tortoise   & 9        & 325.62        & 18         & 1194.34          & 19       & 30.27         & \textbf{19} & \textbf{9.77}    \\
            econ-beause                    & 8        & 657.02        & 19         & 9.35             & 19       & 9.74          & \textbf{19} & \textbf{1.20}    \\
            bio-diseasome                  & 11       & 808.09        & 19         & 37.80            & 19       & 9.28          & \textbf{19} & \textbf{2.47}    \\
            soc-wiki-Vote                  & 11       & 582.12        & 19         & 897.86           & 19       & 29.54         & \textbf{19} & \textbf{3.45}    \\
            ca-CSphd                       & 6        & 360.66        & 19         & 43.00            & 19       & 64.20         & \textbf{19} & \textbf{6.93}    \\
            inf-euroroad                   & 4        & 137.41        & 13         & 681.26           & 19       & 747.87        & \textbf{19} & \textbf{285.02}  \\
            email-univ                     & 9        & 282.53        & 8          & 600.03           & 19       & 1580.24       & \textbf{19} & \textbf{698.38}  \\
            econ-mahindas                  & 9        & 489.03        & 9          & 1235.96          & 19       & 59.13         & \textbf{19} & \textbf{14.12}   \\
            bio-yeast                      & 9        & 246.89        & 18         & 1414.46          & 19       & 173.75        & \textbf{19} & \textbf{40.10}   \\
            comsol                         & 2        & 45.50         & 13         & 233.12           & 19       & 33.33         & \textbf{19} & \textbf{24.88}   \\
            medulla\_1  & 19       & 41.25         & 19         & 120.37           & 19       & 24.75         & \textbf{19} & \textbf{3.78}    \\
            rt\_assad                      & 11       & 726.40        & 19         & 104.49           & 19       & 49.99         & \textbf{19} & \textbf{3.54}    \\
            bio-WormNet-v3       & 5        & 632.09        & 11         & 401.14           & 19       & 48.85         & \textbf{19} & \textbf{37.33}   \\
            heart2                         & 2        & 252.99        & 10         & 212.30           & 19       & 33.35         & \textbf{19} & \textbf{16.66}   \\
            econ-orani678                  & 12       & 1494.00       & 10         & 652.86           & 19       & 68.09         & \textbf{19} & \textbf{16.49}   \\
            road-minnesota                 & 2        & 274.75        & \textbf{8} & \textbf{2606.73} & -        & -             & 2           & 681.46           \\
            inf-openflights                & 10       & 544.23        & 9          & 1070.32          & 19       & 333.10        & \textbf{19} & \textbf{53.83}   \\
            web-edu                        & 19       & 208.18        & 19         & 46.54            & 19       & 33.03         & \textbf{19} & \textbf{4.70}    \\
            econ-psmigr3                   & 8        & 292.08        & 6          & 489.65           & 19       & 1215.21       & \textbf{19} & \textbf{779.61}  \\
            fb-pages-tvshow                & 7        & 351.56        & 12         & 1750.34          & 19       & 1068.56       & \textbf{19} & \textbf{246.12}  \\
            ca-GrQc                        & 6        & 357.37        & 8          & 756.79           & 19       & 4482.46       & \textbf{19} & \textbf{1409.26} \\
            inf-power                      & 2        & 91.74         & 12         & 5160.91          & 16       & 3048.41       & \textbf{18} & \textbf{1490.94} \\
            ca-Erdos992                    & 6        & 422.38        & 10         & 606.08           & 19       & 865.37        & \textbf{19} & \textbf{265.08}  \\
            soc-advogato                   & 10       & 1267.51       & 12         & 1882.45          & 19       & 489.18        & \textbf{19} & \textbf{94.97}   \\
            fb-pages-politician            & 6        & 676.17        & 9          & 2611.87          & 19       & 1703.48       & \textbf{19} & \textbf{358.77}  \\
            ia-reality                     & 8        & 1172.90       & 19         & 161.17           & 19       & 48.34         & \textbf{19} & \textbf{4.52}    \\
            bio-dmela                      & 7        & 1532.54       & 4          & 1670.15          & 2        & 1094.27       & \textbf{12} & \textbf{2582.48} \\
            rt\_onedirection               & 5        & 880.90        & 19         & 132.52           & 19       & 28.15         & \textbf{19} & \textbf{2.61}    \\
            bio-grid-human                 & 6        & 1770.32       & -          & -                & 19       & 6525.33       & \textbf{19} & \textbf{700.97}  \\
            rt\_lolgop                     & 19       & 6029.30       & 19         & 192.43           & 19       & 29.90         & \textbf{19} & \textbf{4.73}    \\
            ia-escorts-dynamic             & 6        & 2483.43       & -          & -                & 9        & 4267.95       & \textbf{19} & \textbf{1915.29} \\
            ca-HepPh                       & 4        & 1944.25       & 3          & 1693.44          & -        & -             & \textbf{5}  & \textbf{1627.93} \\
            web-indochina-2004             & 1        & 498.81        & 5          & 1098.12          & 19       & 389.42        & \textbf{19} & \textbf{177.32}  \\
            soc-anybeat                    & -        & -             & 19         & 7603.85          & 19       & 659.06        & \textbf{19} & \textbf{125.06}  \\
            fb-pages-company               & -        & -             & -          & -                & -        & -             & \textbf{15} & \textbf{3474.49} \\
            econ-poli-large                & -        & -             & 16         & 4605.45          & 19       & 416.15        & \textbf{19} & \textbf{72.00}   \\
            ca-AstroPh                     & -        & -             & -          & -                & -        & -             & -           & -                \\
            ca-CondMat                     & -        & -             & -          & -                & -        & -             & \textbf{4}  & \textbf{1727.81} \\
            ca-cit-HepTh                   & -        & -             & -          & -                & -        & -             & \textbf{17} & \textbf{4769.13} \\
            fb-pages-media                 & -        & -             & -          & -                & -        & -             & \textbf{2}  & \textbf{1088.01} \\
            soc-gemsec-RO                  & -        & -             & -          & -                & -        & -             & -           & -                \\
            soc-gemsec-HU                  & -        & -             & -          & -                & -        & -             & -           & -                \\
            fb-pages-artist                & -        & -             & -          & -                & -        & -             & -           & -                \\
            \hline
        \end{tabular}
    }
    \label{tab:sota-comparison}
\end{table}

    \clearpage
    \begin{table}
    \centering
    \caption{
        Per-graph comparison of \textsc{Grover} and its different configurations.
        Bold entries indicate the configuration that solved the most instances and in case of a tie, achieved the lowest total runtime.
        \textsc{opt. Grover} is not included in this comparison.
    }
    \resizebox{\linewidth}{!}{%
        \begin{tabular}{lrrrrrrrr|rr}
            \hline
            \multirow{2}{*}{Graph} & \multicolumn{2}{c}{\textsc{Base}} & \multicolumn{2}{c}{\textsc{Base + AV}} & \multicolumn{2}{c}{\textsc{Base + $\widetilde{d}(v)$}} & \multicolumn{2}{c}{\textsc{Grover}} & \multicolumn{2}{c}{\textsc{opt. Grover}} \\
            \cline{2-11}
            & \#Solved & $\sum$Seconds & \#Solved & $\sum$Seconds & \#Solved & $\sum$Seconds & \#Solved & $\sum$Seconds & \#Solved & $\sum$Seconds \\
            \hline
            contiguous-usa & 19 & 0.18 & 19 & 0.19 & 19 & 0.18 & \textbf{19} & \textbf{0.18} & 19 & 0.16 \\
            brain\_1 & 19 & 0.18 & \textbf{19} & \textbf{0.18} & 19 & 0.20 & 19 & 0.21 & 19 & 0.21 \\
            arenas-jazz & 19 & 0.58 & \textbf{19} & \textbf{0.52} & 19 & 0.53 & 19 & 0.63 & 19 & 0.40 \\
            ca-netscience & 19 & 1.67 & 19 & 1.08 & 19 & 1.34 & \textbf{19} & \textbf{1.03} & 19 & 0.79 \\
            robot24c1\_mat5 & \textbf{19} & \textbf{5.55} & 19 & 5.57 & 19 & 7.33 & 19 & 7.20 & 19 & 2.99 \\
            tortoise & 19 & 22.84 & 19 & 13.91 & 19 & 15.55 & \textbf{19} & \textbf{9.77} & 19 & 7.65 \\
            econ-beause & \textbf{19} & \textbf{0.59} & 19 & 0.62 & 19 & 1.15 & 19 & 1.20 & 19 & 0.70 \\
            bio-diseasome & 19 & 5.02 & \textbf{19} & \textbf{2.44} & 19 & 3.89 & 19 & 2.47 & 19 & 1.68 \\
            soc-wiki-Vote & 19 & 17.31 & 19 & 8.22 & 19 & 8.06 & \textbf{19} & \textbf{3.45} & 19 & 2.76 \\
            ca-CSphd & 19 & 43.84 & 19 & 13.42 & 19 & 25.26 & \textbf{19} & \textbf{6.93} & 19 & 6.15 \\
            inf-euroroad & 19 & 633.12 & 19 & 526.58 & 19 & 325.68 & \textbf{19} & \textbf{285.02} & 19 & 292.31 \\
            email-univ & 19 & 1596.33 & 19 & 736.96 & 19 & 1285.73 & \textbf{19} & \textbf{698.38} & 19 & 365.96 \\
            econ-mahindas & 19 & 21.32 & 19 & 21.32 & \textbf{19} & \textbf{13.95} & 19 & 14.12 & 19 & 8.51 \\
            bio-yeast & 19 & 150.92 & 19 & 79.70 & 19 & 70.81 & \textbf{19} & \textbf{40.10} & 19 & 32.03 \\
            comsol & \textbf{19} & \textbf{20.76} & 19 & 20.80 & 19 & 25.50 & 19 & 24.88 & 19 & 10.14 \\
            medulla\_1 & 19 & 9.99 & 19 & 5.27 & 19 & 6.29 & \textbf{19} & \textbf{3.78} & 19 & 2.39 \\
            rt\_assad & 19 & 26.74 & 19 & 7.23 & 19 & 12.33 & \textbf{19} & \textbf{3.54} & 19 & 2.72 \\
            bio-WormNet-v3 & 19 & 21.84 & \textbf{19} & \textbf{16.00} & 19 & 38.27 & 19 & 37.33 & 19 & 6.07 \\
            heart2 & \textbf{19} & \textbf{6.41} & 19 & 6.73 & 19 & 16.47 & 19 & 16.66 & 19 & 6.10 \\
            econ-orani678 & \textbf{19} & \textbf{16.14} & 19 & 16.24 & 19 & 16.44 & 19 & 16.49 & 19 & 10.22 \\
            road-minnesota & - & - & 2 & 1091.03 & 1 & 289.35 & \textbf{2} & \textbf{681.46} & 2 & 630.75 \\
            inf-openflights & 19 & 248.31 & 19 & 165.27 & 19 & 66.68 & \textbf{19} & \textbf{53.83} & 19 & 39.03 \\
            web-edu & 19 & 12.02 & \textbf{19} & \textbf{1.95} & 19 & 8.13 & 19 & 4.70 & 19 & 0.81 \\
            econ-psmigr3 & 19 & 942.70 & 19 & 941.32 & 19 & 792.62 & \textbf{19} & \textbf{779.61} & 19 & 516.21 \\
            fb-pages-tvshow & 19 & 831.11 & 19 & 519.97 & 19 & 350.16 & \textbf{19} & \textbf{246.12} & 19 & 202.62 \\
            ca-GrQc & 19 & 4205.63 & 19 & 2528.51 & 19 & 1903.46 & \textbf{19} & \textbf{1409.26} & 19 & 1002.81 \\
            inf-power & 16 & 2614.84 & 17 & 2195.14 & 17 & 1819.69 & \textbf{18} & \textbf{1490.94} & 18 & 1367.21 \\
            ca-Erdos992 & 19 & 956.03 & 19 & 419.74 & 19 & 574.49 & \textbf{19} & \textbf{265.08} & 19 & 224.66 \\
            soc-advogato & 19 & 280.71 & 19 & 165.65 & 19 & 114.31 & \textbf{19} & \textbf{94.97} & 19 & 86.78 \\
            fb-pages-politician & 19 & 1103.28 & 19 & 809.97 & 19 & 391.78 & \textbf{19} & \textbf{358.77} & 19 & 319.89 \\
            ia-reality & 19 & 24.71 & \textbf{19} & \textbf{3.60} & 19 & 20.28 & 19 & 4.52 & 19 & 2.22 \\
            bio-dmela & 2 & 1058.13 & 8 & 3491.01 & 11 & 2891.56 & \textbf{12} & \textbf{2582.48} & 13 & 2289.61 \\
            rt\_onedirection & 19 & 12.14 & \textbf{19} & \textbf{1.23} & 19 & 10.10 & 19 & 2.61 & 19 & 0.86 \\
            bio-grid-human & 19 & 4382.76 & 19 & 3414.61 & 19 & 1413.84 & \textbf{19} & \textbf{700.97} & 19 & 568.92 \\
            rt\_lolgop & 19 & 15.81 & \textbf{19} & \textbf{2.11} & 19 & 16.09 & 19 & 4.73 & 19 & 1.33 \\
            ia-escorts-dynamic & 18 & 7058.83 & 19 & 4815.43 & 19 & 2035.28 & \textbf{19} & \textbf{1915.29} & 19 & 1841.93 \\
            ca-HepPh & - & - & - & - & 3 & 828.59 & \textbf{5} & \textbf{1627.93} & 5 & 1362.08 \\
            web-indochina-2004 & 19 & 346.68 & \textbf{19} & \textbf{143.52} & 19 & 247.44 & 19 & 177.32 & 19 & 52.91 \\
            soc-anybeat & 19 & 496.25 & 19 & 286.42 & 19 & 161.94 & \textbf{19} & \textbf{125.06} & 19 & 103.91 \\
            fb-pages-company & - & - & 2 & 996.48 & 15 & 4061.42 & \textbf{15} & \textbf{3474.49} & 16 & 3403.25 \\
            econ-poli-large & 19 & 340.90 & 19 & 89.17 & 19 & 226.82 & \textbf{19} & \textbf{72.00} & 19 & 41.11 \\
            ca-AstroPh & - & - & - & - & - & - & - & - & - & - \\
            ca-CondMat & - & - & - & - & 2 & 1084.54 & \textbf{4} & \textbf{1727.81} & 4 & 1688.73 \\
            ca-cit-HepTh & - & - & - & - & 17 & 5801.30 & \textbf{17} & \textbf{4769.13} & 18 & 4096.97 \\
            fb-pages-media & - & - & - & - & - & - & \textbf{2} & \textbf{1088.01} & 2 & 1034.82 \\
            soc-gemsec-RO & - & - & - & - & - & - & - & - & - & - \\
            soc-gemsec-HU & - & - & - & - & - & - & - & - & - & - \\
            fb-pages-artist & - & - & - & - & - & - & - & - & - & - \\
            \hline
        \end{tabular}
    }
    \label{tab:config-comparison}
\end{table}

    \clearpage

\end{document}